\documentclass{IEEEojcsys}

\jvol{00}
\jnum{00}
\paper{0000000}
\pubyear{2026}
\receiveddate{29 May 2026}
\reviseddate{28 August 2026}
\accepteddate{13 September 2026}
\publisheddate{XX Month 2026}
\currentdate{XX Month 2026}
\doiinfo{OJCSYS.2026.Doi Number}
\sptitle{Regular Paper}
    
    \usepackage{jhoagg}
    \usepackage{mathtools}
    \usepackage{amsfonts}
    \usepackage{amssymb,latexsym}
    \usepackage[psamsfonts]{eucal}
    \usepackage{amsthm}
    \usepackage{graphicx}
    \let\labelindent\relax
    \usepackage[inline]{enumitem}
    \usepackage{indentfirst}
    \usepackage{setspace}
    \usepackage{microtype}
    \usepackage{threeparttable}
    \usepackage{float}
    \usepackage[colorlinks,urlcolor=blue,linkcolor=blue,citecolor=blue]{hyperref}
    \usepackage{cleveref}
    \usepackage{afterpage}
    \usepackage{placeins}
    \usepackage{cancel}
    \usepackage{cite}
    \usepackage{leftidx}
    \usepackage{breqn}
    \usepackage[export]{adjustbox}
    \usepackage{comment}
    \usepackage[linesnumbered,ruled,vlined]{algorithm2e} 
    \usepackage[dvipsnames]{xcolor}
    \usepackage{soul}
    \usepackage{siunitx}
    \usepackage{arydshln}

    \usepackage{tikz}
    
    \let\originalleft\left
    \let\originalright\right
    \renewcommand{\left}{\mathopen{}\mathclose\bgroup\originalleft}
    \renewcommand{\right}{\aftergroup\egroup\originalright}

    \newcounter{thm} 
    \newtheorem{theorem}[thm]{\indent Theorem}
    
    \newtheorem{assumption}{\indent Assumption}
    
    \newtheorem{proposition}{\indent Proposition}
    
    \newtheorem{lemma}{\indent Lemma}
    
    \newtheorem{corollary}{\indent Corollary}
    
    \newtheorem{definition}{\indent Definition}

    \newtheorem{problem}{\indent Problem}
    \newtheorem{example}{\indent Example}

    \newtheorem{Simulation}{Simulation}

    \newtheorem{fact}{\indent Fact}
    
    \newtheorem{conjecture}{\indent Conjecture}
    
    \newtheorem{experiment}{\indent Experiment}
    
    \newtheorem{remark}{\indent Remark}
    \allowdisplaybreaks

    \renewcommand{\theenumi}{{\it (\alph{enumi})}}
    \renewcommand{\labelenumi}{\theenumi}

    		\newcommand\xqed[1]{%
      \leavevmode\unskip\penalty9999 \hbox{}\nobreak\hfill
      \quad\hbox{#1}}
    \newcommand\exampletriangle{\xqed{$\blacktriangle$}}

    \usepackage{accents}
    
    \newlength\figureheight 
    \newlength\figurewidth

    \allowdisplaybreaks
    
    \graphicspath{ {Figures/} }
    \DeclareGraphicsExtensions{.pdf,.png,.jpg}
    
    \DeclareMathAlphabet{\mathcal}{OMS}{cmsy}{m}{n} % Use the standard calligraphy font

    \crefname{equation}{}{}
    \SetKwInput{KwInit}{Initialize}
    \SetKwFunction{SafeSet}{SafeSet}
    \SetKwFunction{GetSafe}{GetSafe}
    \SetKwFunction{GetUnsafe}{GetUnsafe}
    \SetCommentSty{mycommfont}

    \newcommand{\bd}[0]{\mbox{bd }}

\begin{document}

\title{Predicted-Flow Control Barrier Functions for Real-Time Safe Optimal Control}

\author{AMIRSAEID SAFARI\affilmark{1}}
\author{JESSE B. HOAGG\affilmark{1} (Senior Member, IEEE)}

\affil{Department of Mechanical and Aerospace Engineering, University of Kentucky, Lexington, KY 40506 USA}

\corresp{CORRESPONDING AUTHOR: A. Safari (e-mail: \href{mailto:amirsaeid.safari@uky.edu}{amirsaeid.safari@uky.edu})}
\authornote{This work is supported in part by the National Science Foundation (2450718, 2531196) and USDA NIFA (2024-69014-42393). \\ The source code and visual summary supporting this article are publicly available at \url{https://amirsaeid254.github.io/FlowBarrier/}.}

\editor{This paper was recommended by Senior Editor P. Seiler.}

\markboth{Predicted-Flow Control Barrier Functions for Real-Time Safe Optimal Control}{A. SAFARI AND J. B. HOAGG}

\begin{abstract}
Control barrier functions (CBFs) provide real-time safety guarantees through conditions enforced pointwise in time.
However, synthesizing a valid CBF is difficult and controllers obtained from pointwise CBF-based optimization are typically myopic.
To address myopia, this article introduces predicted-flow control barrier functions (P-CBFs), which generalize the CBF concept from a function of the current state to a functional of a predicted flow under a parametrized control plan over a finite prediction horizon.
For safety, a P-CBF can certify that the predicted flow is in a safe set over the entire prediction horizon.
However, candidate P-CBFs suffer from the same challenge as candidate CBFs, namely, control constraints make it difficult to guarantee that the P-CBF is valid.
This article resolves the validity challenge by introducing a terminal candidate P-CBF requiring that the predicted flow end in a backup safe set at the terminal time, and a planning-time shift that modulates the prediction horizon, providing an additional degree of freedom to ensure feasibility.
Then, the real-time control and the evolution of the control-plan parameter and planning-time shift are determined jointly by a single convex optimization that is guaranteed to be feasible provided a backup safe set is known. 
This safe optimal flow control not only guarantees safety of the state trajectory but also provides a safety certificate of the predicted flow over the entire prediction horizon at each instant of real time. 
The method unifies finite-horizon integral-cost optimization with safety certification.
This optimization reduces to a quadratic program (QP) if the control constraints are a convex polyhedron.
The QP implementation, termed \textit{FlowBarrier}, is validated on a nonholonomic ground robot navigating a dense environment.
FlowBarrier is compared to nonlinear model predictive control and two CBF-based safety filter methods across 100 trials, where FlowBarrier achieves the highest goal-reaching rate, zero safety violations, and the lowest computation time.
\end{abstract}
\begin{IEEEkeywords}
Optimal control, predictive control for nonlinear systems, constrained control, robotics.
\end{IEEEkeywords}

\maketitle

\section{INTRODUCTION}

Autonomous robots are of interest for real-time navigation in proximity to obstacles, humans, and other robots.
Applications include aerial inspection, warehouse logistics, and autonomous mobility \cite{ollero2022aerial,wurman2008coordinating,paden2016survey}.
These applications require that robots achieve performance objectives (e.g., way-point navigation, coordination, formation) while maintaining safety and respecting control input limits (e.g., actuator saturation). 
Safety can be formalized as forward invariance of a prescribed safe set $\SC_\rms \subseteq \BBR^n$ \cite{blanchini1999}.
Frameworks for enforcing forward invariance include Hamilton-Jacobi reachability analysis \cite{bansal2017hj,chen2018hj,herbert2021}, model predictive control \cite{mayne2000,rawlings2017,koller2018}, and barrier functions \cite{wieland2007,ames2016control,safari2024ACC,safari2024TSCT,safari2025safe,rabiee2024closed,compositionACC}.

Control barrier functions (CBFs) provide techniques for selecting controls that enforce forward invariance.
For a control-affine system, the CBF condition is affine in the control.
In this case, a quadratic program (QP) can be used to compute a minimum-intervention control that ensures forward invariance of the CBF's zero-superlevel set \cite{ames2016control}.
Thus, safety enforcement is achieved using a computationally efficient, real-time safety filter that can be implemented in a hierarchical architecture.
However, two limitations constrain applicability of this approach.

First, synthesizing a valid CBF on a large subset of $\SC_\rms$ is challenging due to control input constraints. 
Existing methods based on Hamilton-Jacobi reachability \cite{choi2021hj} and sum-of-squares programming \cite{wang2018} are limited to low-dimensional systems or restricted dynamics.

Second, even if a valid CBF is available, pointwise CBF-based optimization is myopic.
Specifically, the safety certificate is enforced pointwise in time based only on the current state with no consideration of how the trajectory will evolve over a future time horizon.
As a result, the CBF does not account for conflicts between safety and performance that develop over the horizon---this can lead to aggressive corrections, conservative behavior, and poor performance.

Backup-CBF methods have been developed to address the challenge of synthesizing a CBF subject to input limits \cite{gurriet2020,backupACC}.
Instead of relying on a valid CBF on $\SC_\rms$, backup-CBF methods specify a backup controller $u_\rmb$ together with a backup set $\SC_\rmb \subset \SC_\rms$ that is forward invariant under $u_\rmb$.
Then, the implicit control-forward-invariant subset of $\SC_\rms$ is defined as the set of states from which the trajectory under $u_\rmb$ is in $\SC_\rms$ over a finite horizon and ends in $\SC_\rmb$ at the terminal time.
This construction is conservative because the backup controller $u_\rmb$ that makes $\SC_\rmb$ forward invariant is also used to drive the state to $\SC_\rmb$ by the terminal time.
Thus, the implicit control-forward-invariant set is often small relative to the maximal control-forward-invariant subset of $\SC_\rms$.
Recent work has sought to mitigate this conservatism by delaying the switch from the nominal controller to $u_\rmb$ until the last time at which the backup trajectory reaches $\SC_\rmb$ without leaving $\SC_\rms$ \cite{agrawal2025gatekeeper,singletary2022safe}, or by composing multiple backup sets and backup controllers to enlarge the implicit control-forward-invariant subset of $\SC_\rms$ \cite{backupautomatic}.
Nevertheless, these methods still rely on fixed backup controllers to drive the state toward $\SC_\rmb$, and thus, conservatism persists.
A related approach in \cite{freire2025designing} parameterizes the backup controller to obtain a family of backup controllers, where each one 
is designed \textit{a priori} to stabilize one safe closed-loop equilibrium. 
Together, these equilibria create a backup safe manifold, which is analogous to a backup safe set.
Similar to backup-CBF methods, \cite{freire2025designing} can be conservative because the family of backup controls may have a small region of convergence.
Another shortcoming of these approaches is that the prediction horizon is used only for safety certification, not for performance optimization.
Furthermore, the control is determined by a pointwise condition based only on the current state with no mechanism to optimize cost over the prediction horizon.

The myopic nature of pointwise CBF-based optimization has been addressed by incorporating finite-horizon prediction into the safety certificate. 
One approach is to compose a finite-horizon planner with a CBF-based safety filter in a layered architecture \cite{grandia2021layered,sforni2024receding}. 
The planner solves its own optimization with discrete CBF constraints to generate a reference, and the safety filter then enforces barrier constraints on the applied control through the CBF. 
This layered design introduces at least one additional optimization beyond the CBF and does not address the validity of the barrier function. 
Even if the safety filter is replaced by a backup CBF to recover validity, the planner and the backup controller each propagate the system dynamics on a finite horizon for different purposes---one for cost and the other for safety, and these propagations are decoupled with potentially conflicting goals. 
A second approach avoids the layered structure and addresses prediction and safety in a single optimization \cite{breeden2022}. 
Specifically, \cite{breeden2022} propagates a trajectory under a fixed nominal controller and encodes its future safety into a barrier condition that is affine in the current control; however, the certified safe set is determined by the nominal controller, and feasibility under input constraints is not guaranteed. 
The approach in \cite{vahs2024} parametrizes the control trajectory to address invariance in trajectory space; however, the resulting QP is not guaranteed to be feasible, and thus safety is not guaranteed. 
Motivated by embedding CBF constraints in a receding-horizon optimization \cite{zeng2021mpccbf}, several methods improve tractability through iterative linearization \cite{liu2025learning}, sampling-based trajectory optimization with a closed-form CBF filter \cite{rabiee2025guaranteed}, or quadratic approximation of the cost-to-go \cite{gutierrez2026receding}; however, the validity of the barrier function at each step is assumed rather than established, and the underlying pointwise CBF condition can become infeasible under input constraints.
In summary, existing methods that incorporate prediction into the safety certificate either decouple prediction from certification across separate optimizations, or rely on a single optimization whose feasibility under input constraints is not guaranteed.

This article presents a non-myopic real-time optimal control for simultaneous safety and performance with guaranteed feasibility in the presence of input constraints.
First, the article introduces predicted-flow control barrier functions (P-CBFs), which can be used to certify safety of a predicted flow over a finite prediction horizon.
The P-CBF generalizes the CBF concept from a function of the current state to a functional of the predicted flow under a parametrized control plan.
Specifically, the control plan is parametrized by a finite-dimensional variable $\theta$ over a prediction horizon $T$, and the dynamics are propagated over this horizon under the control plan to obtain the predicted flow $\varphi$.
A logical candidate P-CBF is the minimum of $h_\rms(\varphi)$ over the prediction horizon, where the safe set $\SC_\rms$ is the zero-superlevel set of $h_\rms$, which is not assumed to be a valid CBF. 
However, this candidate P-CBF suffers from the same challenge as standard candidate CBFs---namely, control constraints make it difficult to guarantee and verify that the set of controls satisfying the P-CBF condition is nonempty.
A natural remedy is to add a terminal candidate P-CBF that requires the predicted flow to end in a backup safe set $\SC_\rmb$ at the terminal time. 
Still, it remains difficult to guarantee that this candidate P-CBF pair is a valid P-CBF pair under input constraints.

\begin{figure}[t]
\centering
\includegraphics[width=0.9\columnwidth, clip=true, trim=3.85in 0.8in 2.90in 2.50in]{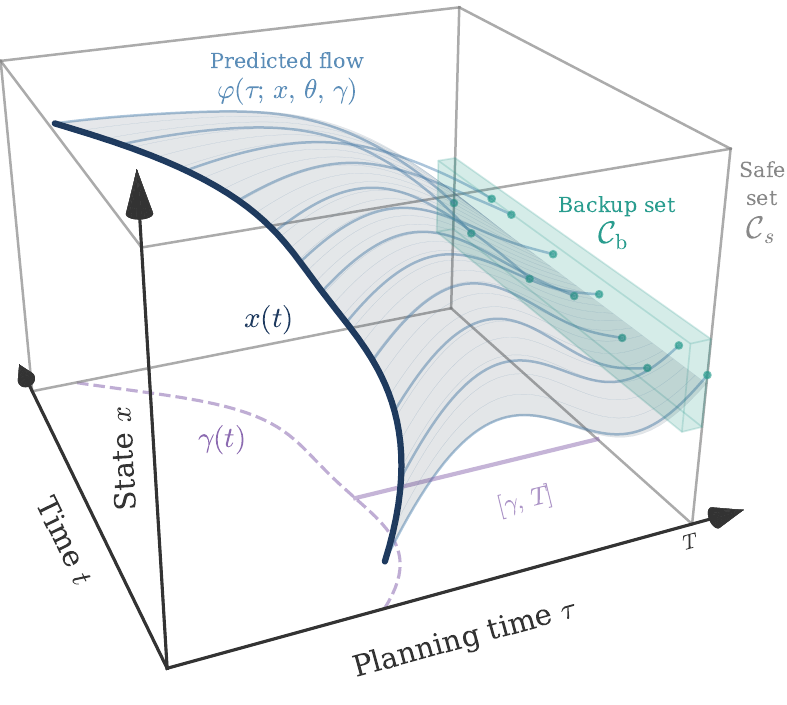}
\caption{The schematic illustrates the safe set $\SC_\rms$, the backup set $\SC_\rmb$, the closed-loop trajectory $x(t)$, and the predicted flows $\varphi(\tau; x, \theta, \gamma)$ along the planning time $\tau$. At each time $t$, the predicted flow propagates over the prediction window $[\gamma(t), T]$, remains within $\SC_\rms$, and terminates in $\SC_\rmb$ at $\tau = T$. The planning-time shift $\gamma(t)$ modulates the prediction window, providing the degree of freedom that guarantees feasibility.}
\label{fig:flowbarrier_overview}
\end{figure}

To address validity/feasibility, this article introduces a scalar planning-time shift $\gamma$ that modulates the prediction window, providing an additional degree of freedom in the optimization (see \Cref{fig:flowbarrier_overview}). 
The planning-time shift $\gamma$ and control-plan parameter $\theta$ are treated as dynamic states and combined with the system state $x$ to form an augmented state $(x,\theta,\gamma)$. 
The applied control $u$, update of $\theta$, and evolution of $\gamma$ are then determined jointly as the solution of a single convex optimization that is guaranteed to be feasible at every time provided a backup safe set is known. 
This backup safe set is analogous to the backup set in backup-CBF methods \cite{gurriet2020,backupACC} and the terminal constraint in model predictive control with guaranteed recursive feasibility \cite{mayne2000}. 
The resulting safe optimal flow control provides a safety certificate over the entire prediction horizon and unifies cost optimization with safety certification in a single convex optimization whose feasibility under input constraints is guaranteed.
If control constraints are a convex polyhedron (i.e., an $m$-dimensional polyhedron), then the optimization reduces to a QP, named \textit{FlowBarrier}.

This article is organized as follows. 
\Cref{sec:directional_derivatives} reviews directional derivatives, and \Cref{sec:dcbf} introduces directional-derivative CBFs, which extend the idea of CBFs to functions that are directionally differentiable but not necessarily continuously differentiable.
These directional-derivative CBFs are used in \Cref{sec:pcbf} to introduce P-CBFs and show that P-CBFs can be used to obtain forward invariance in trajectory space. 
\Cref{sec:problem} formulates the safe optimal control problem addressed in this article.
Then, \Cref{sec:control_forward_invariant} presents the safe optimal flow control solution with guaranteed feasibility, and 
\Cref{sec:impl_remark} presents the QP implementation, named FlowBarrier.
\Cref{sec:application} applies FlowBarrier to a nonholonomic ground robot and compares the algorithm to nonlinear model predictive control and two CBF-based safety filters paired with an iterative linear-quadratic regulator planner across 100 trials, where FlowBarrier achieves the highest goal-reaching rate, zero safety violations, and the lowest computation time.
The source code for FlowBarrier and all comparison methods is publicly available in \texttt{CBFJAX}, an open-source library developed alongside this work that provides automatic differentiation, just-in-time compilation, and a unified benchmarking environment for safe optimal control methods.

\section{DIRECTIONAL DERIVATIVES}
\label{sec:directional_derivatives}

Let $\SD \subseteq \BBR^n$, and let $\mu \colon \SD \to \BBR$ be continuous.
The \textit{radial cone} $R_\SD \colon \SD \rightrightarrows \BBR^n$ is defined by
\begin{align*}
R_\SD(x) &\triangleq \{ \nu \in \BBR^n \colon \exists \,  \varepsilon > 0 \text{ such that } \\
&\qquad \forall\, s \in (0,\varepsilon),\, x + s\nu \in \SD \}.
\end{align*}
Note that if $x \in \operatorname{int}\SD$, then $R_\SD(x) = \BBR^n$.
The function $\mu$ is \textit{right-side directionally differentiable on $\SD$} if for all $(x,\nu) \in \SD \times R_\SD(x)$,
\begin{equation}\label{eq:rh_directional_derivative}
D_{\nu}  \mu(x) \triangleq \lim_{s \downarrow 0} \frac{\mu(x + s \nu) - \mu(x)}{s}
\end{equation}
exists.
If $\mu$ is differentiable on $\SD$, then $D_{\nu} \mu(x) = L_{\nu} \mu(x)$, where $L_\nu \mu(x) \triangleq \mu^\prime(x) \nu$ is the Lie derivative of $\mu$ along $\nu$.
This article is not concerned with functions that are left-side directionally differentiable (i.e., $\lim_{s \uparrow 0}$) but not differentiable.
Thus, for brevity, we omit ``right-side'' for the remainder of the article.

The next result concerns the time derivative of $\mu(y(t))$ from the right side.

\begin{lemma}\label{lem:dini_directional}
\rm
Assume $\mu$ is locally Lipschitz and directionally differentiable on $\SD$.
Let $y \colon [0,\infty) \to \SD$ be differentiable such that for all $t \geq 0$, $\dot{y}(t) \in R_\SD(y(t))$.
Then, for all $t \geq 0$,
\begin{equation*}
\frac{\rmd^+}{\rmd t} \mu(y(t)) \triangleq \lim_{s \downarrow 0} \frac{\mu(y(t+s)) - \mu(y(t))}{s}
\end{equation*}
exists, and
\begin{equation}\label{eq:dini_equals_directional}
\frac{\rmd^+}{\rmd t} \mu(y(t)) = D_{\dot{y}(t)} \mu(y(t)).
\end{equation}
\end{lemma}

\begin{proof}
\indent Since $y$ is differentiable, it follows from the Taylor expansion that there exists $o \colon [0,\infty) \to \BBR^n$ such that $y(t+s) = y(t) + s \dot{y}(t) + o(s)$ and $\lim_{s \downarrow 0} \|o(s)\|/s = 0$.
Note that
\begin{equation}\label{eq:quotient_split}
\frac{\mu(y(t+s)) - \mu(y(t))}{s} = \frac{\mu(y(t) + s\dot{y}(t)) - \mu(y(t))}{s} + \eta(s),
\end{equation}
where
\begin{equation*}
\eta(s) \triangleq \frac{\mu(y(t) + s\dot{y}(t) + o(s)) - \mu(y(t) + s\dot{y}(t))}{s}.
\end{equation*}
Since $\mu$ is locally Lipschitz, there exists $M > 0$ such that $|\eta(s)| \leq M\|o(s)\|/s$, which implies that $\lim_{s \downarrow 0} \eta(s) = 0$.
Thus, taking the limit of \eqref{eq:quotient_split} and using \eqref{eq:rh_directional_derivative} yields
\begin{align*}
\frac{\rmd^+}{\rmd t} \mu(y(t))
&= \lim_{s \downarrow 0} \left [ \frac{\mu(y(t) + s\dot{y}(t)) - \mu(y(t))}{s} + \eta(s) \right ]\\
&= D_{\dot{y}(t)} \mu(y(t)),
\end{align*}
which confirms that $\frac{\rmd^+}{\rmd t} \mu(y(t))$ exists and is given by \eqref{eq:dini_equals_directional}.
\end{proof}

\section{DIRECTIONAL-DERIVATIVE CONTROL BARRIER FUNCTIONS}\label{sec:dcbf}

Consider
\begin{equation}\label{eq:affine_control}
    \dot{x}(t) = f(x(t)) + g(x(t)) u(t),
\end{equation}
where $f\colon \BBR^n \to \BBR^n$ and $g\colon \BBR^n \to \BBR^{n \times m}$
are continuously differentiable on $\BBR^n$,
$x(t) \in \BBR^n$ is the state,
$x(0) = x_0 \in \BBR^n$ is the initial condition,
and $u(t) \in \SU \subseteq \BBR^m$ is the control.
The control $u$ is an \textit{admissible control} if for all $t \ge 0$, $u(t) \in \SU$.
Each solution to \eqref{eq:affine_control} that appears in this article is assumed to exist and be unique on $[0,\infty)$.
For notational convenience, we define
\begin{equation*}
F(x,u) \triangleq f(x) + g(x)u,
\end{equation*}
which is the right-hand side of \eqref{eq:affine_control}.

Let $h \colon \BBR^n \to \BBR$ be continuous, and define the zero-superlevel set
\begin{equation*}
\SC \triangleq \{ x \in \BBR^n \colon h(x) \ge 0 \},
\end{equation*}
which is assumed to be nonempty and contain no isolated points.

Let $u_{\rm fi} \colon \SC \to \SU$.
Then, $\SC$ is \textit{forward invariant} with respect to \eqref{eq:affine_control} with $u=u_{\rm fi}$ if for all $x_0 \in \SC$, the solution to \eqref{eq:affine_control} with $u=u_{\rm fi}$ is such that for all $t \in [0,\infty)$, $x(t) \in \SC$.

A continuous function $a \colon \BBR \to \BBR$ is an \textit{extended class-$\SK$ function} if it is strictly increasing and $a(0)=0$.

\begin{definition}\label{def:cbf}
\rm
Assume $h$ is continuously differentiable on $\SC$. 
Then, $h$ is a \textit{control barrier function} (CBF) for
\eqref{eq:affine_control} on $\SC$ if there exists an extended
class-$\SK$ function $\alpha$ such that for all $x \in \SC$,
\begin{equation}\label{eq:cbf_condition}
    \sup_{\hat{u} \in \SU} \; L_f h(x) + L_g h(x)\hat{u}
    + \alpha(h(x)) \ge 0.
\end{equation}
\end{definition}

Definition~\ref{def:cbf} requires that $h$ is continuously differentiable.
The next definition extends the concept of CBF to functions that are directionally differentiable but not necessarily differentiable.

\begin{definition}\label{def:cbf_directional}
\rm
Assume $h$ is locally Lipschitz and directionally differentiable on $\SC$.
Then, $h$ is a \textit{directional-derivative control barrier function} (D-CBF) for \eqref{eq:affine_control} on $\SC$ if there exists an extended class-$\SK$ function $\alpha$ such that for all $x \in \SC$,
\begin{equation}\label{eq:dcbf_condition}
    \sup_{\hat{u} \in \SU} \; D_{F(x,\hat{u})} h(x) + \alpha(h(x)) \ge 0.
\end{equation}
\end{definition}

If $h$ is continuously differentiable, then $D_{f(x)+g(x)\hat{u}} h(x) = L_f h(x) + L_g h(x) \hat{u}$.
In this case, \eqref{eq:dcbf_condition} is equivalent to
\eqref{eq:cbf_condition}, and Definition~\ref{def:cbf_directional} reduces to Definition~\ref{def:cbf}.
The concept of CBF can be further generalized using Dini derivatives (see \cite{wiltz2025predictive}); however, Definition~\ref{def:cbf_directional} suffices for this article.

Next, let $h_1,\ldots,h_\ell \colon \BBR^n \to \BBR$ be continuous, and define
\begin{equation*}%\label{eq:C_vector}
\SC_\rmv \triangleq \{ x \in \BBR^n \colon h_1(x) \ge 0, \ldots, h_\ell(x) \ge 0 \},
\end{equation*}
which is the intersection of the zero-superlevel sets of $h_1,\ldots,h_\ell$. 
The next definition extends the concept of D-CBF to address multiple barrier functions and the intersection of their zero-superlevel sets.

\begin{definition}\label{def:vector_dcbf}
\rm
Assume $h_1,\ldots,h_\ell$ are locally Lipschitz and directionally differentiable on $\SC_\rmv$.
Then, $(h_1,\ldots,h_\ell)$ is a \textit{D-CBF $\ell$-tuple} for \eqref{eq:affine_control} on $\SC_\rmv$ if there exist extended class-$\SK$ functions $\alpha_1,\ldots,\alpha_\ell$ such that for all $x \in \SC_\rmv$, $K_\rmv(x)$ is nonempty, where $K_\rmv \colon \SC_\rmv \rightrightarrows \SU$ is defined by
\begin{align}\label{eq:vector_dcbf_K}
K_\rmv(x) &\triangleq \{ \hat{u} \in \SU \colon D_{F(x,\hat{u})} h_1(x) + \alpha_1(h_1(x)) \ge 0, \ldots, \nn \\
&\qquad D_{F(x,\hat{u})} h_\ell(x) + \alpha_\ell(h_\ell(x)) \ge 0 \}.
\end{align}
\end{definition}

In the case where $\ell=1$, $K_\rmv(x)$ is nonempty if and only if \eqref{eq:dcbf_condition} with $h=h_1$ is satisfied. 
Thus, Definition~\ref{def:vector_dcbf} is equivalent to Definition~\ref{def:cbf_directional} in the $\ell=1$ case.

The next result shows that if $(h_1,\ldots,h_\ell)$ is a D-CBF $\ell$-tuple, then any control from $K_\rmv(x)$ makes $\SC_\rmv$ forward invariant.
This result extends the standard CBF result (e.g., \cite[Corollary~2]{ames2016control}) to D-CBFs.

\begin{theorem}
\label{prop:vector_dcbf_invariance}
\rm
Assume $(h_1,\ldots,h_\ell)$ is a D-CBF $\ell$-tuple for \eqref{eq:affine_control} on $\SC_\rmv$, and let $u_{\rm fi} \colon \SC_\rmv \to \SU$ be such that for all $x \in \SC_\rmv$, $u_{\rm fi}(x) \in K_\rmv(x)$.
Then, $\SC_\rmv$ is forward invariant with respect to \eqref{eq:affine_control} with $u = u_{\rm fi}$.
\end{theorem}

\begin{proof}
\indent Let $x_0 \in \SC_\rmv$.
For $i \in \{1,\ldots,\ell\}$, let $\eta_i \colon[0,\infty) \to \BBR$ satisfy $\dot \eta_i(t) = -\alpha_i(\eta_i(t))$, where $\eta_i(0) = h_i(x_0)$. 
Since $h_i(x_0) \geq 0$ and $\alpha_i$ is an extended class-$\SK$ function, it follows that for all $t \ge 0$, $\eta_i(t) \ge 0$. 

Since $(h_1,\ldots,h_\ell)$ is a D-CBF $\ell$-tuple, Definition~\ref{def:vector_dcbf} implies that for all $x \in \SC_\rmv$, $K_\rmv(x)$ is nonempty.
Thus, Lemma~\ref{lem:dini_directional} and \eqref{eq:vector_dcbf_K} imply that for all $i \in \{1,\ldots,\ell\}$ and all $t \geq 0$,
\begin{align*}
\frac{\rmd^+}{\rmd t} h_i(x(t)) &= D_{F(x(t),u_{\rm fi}(x(t)))} h_i(x(t)) \ge -\alpha_i(h_i(x(t))).
\end{align*}
Hence, the comparison lemma \cite[Lemma~3.4]{khalil2002nonlinear} implies that for all $i \in \{1,\ldots,\ell\}$ and all $t \geq 0$, $h_i(x(t)) \geq \eta_i(t) \ge 0$.
Thus, for all $t \geq 0$, $x(t) \in \SC_\rmv$.
\end{proof}

\section{PREDICTED-FLOW CONTROL BARRIER FUNCTIONS}\label{sec:pcbf}

Let $T > 0$ be the planning-and-prediction horizon, and consider the control plan $u_\rmp(\cdot; \theta) \colon [0,T] \to \BBR^m$, which is parametrized by $\theta \in \BBR^d$.
Hence, $\theta \in \BBR^d$ is a finite-dimensional parameterization of the control plan $u_\rmp(\cdot;\theta)$ on the planning-and-prediction horizon $[0,T]$.
The control plan $u_\rmp$ is continuous on $[0,T] \times \BBR^d$, and for all $\tau \in [0,T]$, $u_\rmp(\tau;\cdot)$ is continuously differentiable on $\BBR^d$.

Let $k \colon \BBR^d \to \BBR$ be continuously differentiable, and define the \textit{admissible parameter set}
\begin{equation}
\Theta \triangleq \{ \theta \in \BBR^d \colon k(\theta) \ge 0 \},\label{eq:theta_set}
\end{equation}
where $\Theta$ is such that for all $(\tau,\theta) \in [0,T] \times \Theta$, $u_\rmp(\tau;\theta) \in \SU$.
Thus, $\theta \in \Theta$ provides a finite-dimensional parameterization of the control plans $u_\rmp(\cdot;\theta)$ that are admissible for the entire horizon $[0,T]$.
The following example provides one construction for $u_\rmp$.

\begin{example}\label{ex:foh}\rm
Let $\beta_1,\ldots,\beta_p \colon [0,T] \to \BBR$ be continuous nonnegative functions such that for all $\tau \in [0,T]$, $\sum_{i=1}^{p} \beta_i(\tau) = 1$. 
Then, consider the control plan 
\begin{equation}\label{eq:u_p_construction}
u_\rmp(\tau;\theta) = \sum_{i=1}^{p} \theta_i \, \beta_i(\tau),
\end{equation}
where $\theta_i \in \BBR^m$, $\theta = [\theta_1^\top \;\theta_2^\top \;\ldots \;\theta_p^\top]^\top \in \BBR^{d}$, and $d = pm$.
For each planning time $\tau$, the control plan \eqref{eq:u_p_construction} is a convex combination of $\theta_1,\ldots,\theta_p$. 
Thus, selecting $k$ such that $\Theta \subseteq \SU^p$$\triangleq \SU \times \cdots \times \SU$ is sufficient to satisfy the condition that for all $(\tau,\theta) \in [0,T] \times \Theta$, $u_\rmp(\tau;\theta) \in \SU$.
If $\SU$ is a convex polyhedron, then it is possible to construct $k$ such that $\Theta \subseteq \SU^p$ and $\Theta$ approximates $\SU^p$.
This construction is provided in Section~\ref{sec:impl_remark}.

One choice for $\beta_1,\ldots,\beta_p$ are degree-one B-splines.
Specifically, let $T_\rmp \triangleq T/(p-1)$, and for $i\in\{1,\ldots,p\}$, let
\begin{equation}\label{eq:u_p_construction.b}
    \beta_i(\tau) = \frac{\tau - (i-2) T_\rmp}{T_\rmp} \sigma_{i-1}(\tau) + \frac{iT_\rmp - \tau}{T_\rmp} \sigma_{i}(\tau),
\end{equation}
where $\sigma_0=0$, and
\begin{equation}\label{eq:u_p_construction.a}
    \sigma_i(\tau) \triangleq
    \begin{cases}
    1, & \mbox{if } \tau\in[ (i-1) T_\rmp, i T_\rmp),\\
    0, & \mbox{else}.
    \end{cases}
\end{equation}
Figure~\ref{fig:basis_functions} illustrates the degree-one B-spline basis functions \cref{eq:u_p_construction.b,eq:u_p_construction.a}, and the resulting control plan \eqref{eq:u_p_construction}.
\exampletriangle
\end{example}

\begin{figure}[t]
\centering
\includegraphics[width=0.85\columnwidth, clip=true, trim=0.0in 0.0in 0.0in 0.0in]{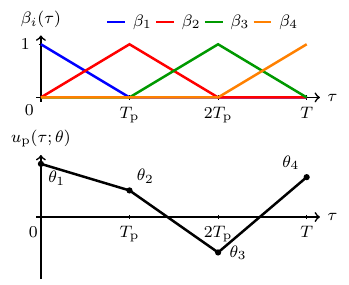}
\caption{First-degree B-spline basis functions $\beta_1,\ldots,\beta_p$ (top) and the resulting control plan $u_\rmp(\tau;\theta) = \sum_{i=1}^p \theta_i \beta_i(\tau)$ (bottom) on $[0,T]$ with $p=4$. At $\tau = (i-1)T_\rmp$, $u_\rmp$ equals $\theta_i$, and $u_\rmp$ is piecewise linear.}
\label{fig:basis_functions}
\end{figure}

The \textit{predicted flow} $\phi(\cdot; x, \theta) \colon [0,T] \to \BBR^n$ satisfies
\begin{equation}\label{eq:flow_def}
    \phi(\tau; x, \theta) = x + \int_{0}^{\tau}
    F\!\left(\phi(\sigma; x, \theta),\, u_\rmp(\sigma; \theta)\right) \, {\rm d}\sigma,
\end{equation}
which implies that $\phi(\tau; x, \theta)$ is the solution to
\eqref{eq:affine_control} at planning time $\tau \in [0,T]$ with initial condition
$x$ and $u = u_\rmp(\cdot; \theta)$.
In other words, $\phi(\cdot; x, \theta)$ is the flow of \eqref{eq:affine_control} from state $x$ under the plan $u_\rmp(\cdot;\theta)$ with parameter $\theta$. 
Thus, the predicted flow depends on the state $x$ and parameter $\theta$, which provides a finite-dimensional representation of $u_\rmp$ over the horizon $[0,T]$.
Differentiating \eqref{eq:flow_def} with respect to $\tau$ yields
\begin{equation*}
\frac{\partial \phi}{\partial \tau}(\tau; x, \theta)
= F\big(\phi(\tau; x, \theta), u_\rmp(\tau; \theta)\big),
\end{equation*}
which is the evolution of the predicted flow $\phi$ given $(x,\theta)$.

At each time $t\ge0$, the predicted flow $\phi(\cdot;x(t),\theta(t))$ 
depends on the current state $x(t)$ and parameter $\theta(t)$. 
The time evolution of $x$ is influenced by the control $u$. 
In order to influence the time evolution of $\theta$, we let $\theta \colon [0,\infty) \to \BBR^d$ be the solution to 
\begin{equation}\label{eq:theta_dyn}
\dot{\theta}(t) = \omega(t),
\end{equation}
where $\theta(0) = \theta_0 \in \BBR^d$ and $\omega \colon [0,\infty) \to \Omega \subseteq \BBR^d$ is the control input to the integrator.
Note that $\Omega = \BBR^d$ implies that there are no limits on the rate of change of the control plan parameters. 
However, $\Omega$ can be selected to limit the rate of change of $\theta$ based on actuator rate limits. 

Next, let
$H_1,\ldots,H_\ell \colon \allowbreak C([0,T], \BBR^n) \to \BBR$
be functionals such that for all $i \in \{1,\ldots,\ell\}$,
\begin{equation}\label{eq:psi_functional}
\psi_i(x,\theta) \triangleq H_i[\phi(\cdot; x, \theta)]
\end{equation}
is locally Lipschitz and directionally differentiable on 
\begin{align}\label{eq:def:Psi}
\Psi &\triangleq \{(x,\theta) \in \BBR^n \times \BBR^d \colon k(\theta) \geq 0,  \psi_1(x,\theta) \geq 0, \ldots, \nn\\
&\qquad \psi_\ell(x,\theta) \geq 0 \}.
\end{align}
Note that $\Psi$ is the set of $(x,\theta)$ such that the predicted flow $\phi$ mapped through each functional $H_i$ is nonnegative, and $\theta \in \Theta$, which implies that $u_\rmp(\tau;\theta) \in \SU$ for the entire planning horizon $\tau\in[0,T]$.

We now introduce the concept of a predicted-flow CBF.
This concept extends the notion of a D-CBF to address the $\ell$-tuple $(\psi_1,\ldots,\psi_\ell)$, where each $\psi_i$ is obtained by mapping $\phi$ through the functional $H_i$ and where $k$ defines the admissible parameters for the control plan.

\begin{definition}\label{def:pfcbf}
\rm
Assume $\psi_1,\ldots,\psi_\ell$ given by \eqref{eq:psi_functional} are locally Lipschitz and directionally differentiable on $\Psi$.
Then, $(\psi_1,\ldots,\psi_\ell)$ is a \textit{predicted-flow control barrier function (P-CBF) $\ell$-tuple for \eqref{eq:affine_control} and \eqref{eq:theta_dyn} on $\Psi$ given $u_\rmp$ and $k$} if there exist extended class-$\SK$ functions $\alpha_1,\ldots,\alpha_\ell,\beta$ such that for all $(x,\theta) \in \Psi$, $K_\Psi(x,\theta)$ is nonempty, where $K_\Psi \colon \Psi \rightrightarrows \SU \times \Omega$ is defined by
\begin{align}\label{eq:pfcbf_K}
K_\Psi(x,\theta) &\triangleq \Big \{(\hat{u}, \hat{\omega}) \in \SU \times \Omega \colon k^\prime(\theta)\,\hat{\omega} + \beta(k(\theta)) \geq 0, \nn\\
&\quad D_{\left[\begin{smallmatrix} F(x,\hat{u}) \\ \hat{\omega} \end{smallmatrix}\right]} \psi_1(x,\theta) + \alpha_1(\psi_1(x,\theta)) \geq 0, \ldots,\nn \\
&\quad D_{\left[\begin{smallmatrix} F(x,\hat{u}) \\ \hat{\omega} \end{smallmatrix}\right]} \psi_\ell(x,\theta) + \alpha_\ell(\psi_\ell(x,\theta)) \geq 0 \Big \}.
\end{align}
\end{definition}

The next result connects P-CBFs to D-CBFs.

\begin{proposition}\label{prop:pcbf_equiv}\rm
    $(\psi_1,\ldots,\psi_\ell)$ is a P-CBF $\ell$-tuple for \eqref{eq:affine_control} and \eqref{eq:theta_dyn} on $\Psi$ if and only if $(k,\psi_1,\ldots,\psi_\ell)$ is a D-CBF $(\ell+1)$-tuple for \eqref{eq:affine_control} and \eqref{eq:theta_dyn} on $\Psi$.
\end{proposition}

\begin{proof}
\indent Since $k$ is continuously differentiable, note that 
$D_{\hat{\omega}} k(\theta) = k^\prime(\theta)\,\hat{\omega}.$
Thus, Definitions \ref{def:vector_dcbf} and \ref{def:pfcbf} imply that $(\psi_1,\ldots,\psi_\ell)$ is a P-CBF $\ell$-tuple on $\Psi$ given $k$ with associate point-to-set map $K_\Psi$ if and only if $(h_1,\ldots,h_\ell,h_{\ell+1})=(k,\psi_1,\ldots,\psi_\ell)$ is a D-CBF $(\ell+1)$-tuple on $\SC_\rmv = \Psi$ with associate map $K_\rmv=K_\Psi$.
\end{proof}

\begin{remark}\rm
If there are no constraints on the control input (i.e., $\SU = \BBR^m$), then $k$ can be selected as a positive constant.
In this case, $k^\prime \hat{\omega} + \beta(k) = \beta(k) > 0$, which implies that the first inequality in \eqref{eq:pfcbf_K} is trivially satisfied.
In this case, $(\psi_1,\ldots,\psi_\ell)$ is a P-CBF $\ell$-tuple on $\Psi$ if and only if $(\psi_1,\ldots,\psi_\ell)$ is a D-CBF $\ell$-tuple on $\Psi$.
\end{remark}

The next result shows that if $(\psi_1,\ldots,\psi_\ell)$ is a P-CBF $\ell$-tuple on $\Psi$, then any control selected pointwise from $K_\Psi(x,\theta)$ makes $\Psi$ forward invariant.
This result is a consequence of Theorem~\ref{prop:vector_dcbf_invariance} and Proposition~\ref{prop:pcbf_equiv}.

\begin{theorem}\label{thm:pfcbf_forward_invariance}\rm
Assume $(\psi_1,\ldots,\psi_\ell)$ is a P-CBF $\ell$-tuple for \eqref{eq:affine_control} and \eqref{eq:theta_dyn} on $\Psi$, and let 
$K_\Psi(x,\theta)$ be given by \eqref{eq:pfcbf_K}. 
Let $u_{\rm fi} \colon \Psi \to \SU$ and $\omega_{\rm fi} \colon \Psi \to \Omega$ be such that for all $(x,\theta) \in \Psi$, $(u_{\rm fi}(x,\theta),\omega_{\rm fi}(x,\theta)) \in K_\Psi(x,\theta)$.
Then, $\Psi$ is forward invariant with respect to \eqref{eq:affine_control} and \eqref{eq:theta_dyn} with $(u,\omega) = (u_{\rm fi} , \omega_{\rm fi})$.
\end{theorem}

\begin{proof}
\indent Since $(\psi_1,\ldots,\psi_\ell)$ is a P-CBF $\ell$-tuple on $\Psi$, Proposition~\ref{prop:pcbf_equiv} implies that $(k,\psi_1,\ldots,\psi_\ell)$ is a D-CBF $(\ell+1)$-tuple on $\Psi$.
Since, in addition, $(u_{\rm fi}(x,\theta),\omega_{\rm fi}(x,\theta)) \in K_\Psi(x,\theta)$, it follows from Theorem~\ref{prop:vector_dcbf_invariance} that $\Psi$ is forward invariant with respect to \eqref{eq:affine_control} and \eqref{eq:theta_dyn} with $(u,\omega) = (u_{\rm fi}, \omega_{\rm fi})$.
\end{proof}

The following subsections present 2 useful candidate P-CBFs.
These candidate P-CBFs are used in subsequent sections of this article.

\subsection{MINIMUM-OVER-PREDICTION-HORIZON P-CBF}

We present a candidate P-CBF for the situation in which it is desirable for the predicted flow $\phi$ to be in a desired set throughout the prediction horizon $[0,T]$.
Specifically, let $h_\rms \colon \BBR^n \to \BBR$ be continuously differentiable, and let the safe set $\SC_\rms$ be given by
\begin{equation}\label{eq:safe_set}
\SC_\rms = \left\{ x \in \BBR^n \colon h_\rms(x) \geq 0 \right\}.
\end{equation}
To determine whether $\phi$ is in $\SC_\rms$ throughout the prediction horizon, consider the candidate P-CBF
\begin{equation}\label{eq:def:psi_m}
\psi_\rmm(x,\theta) = \min_{\tau \in [0,T]} h_\rms(\phi(\tau; x, \theta)).
\end{equation}
Note that $\psi_\rmm$ is nonnegative if and only if $\phi(\tau;x,\theta) \in \SC_\rms$ for all prediction times $\tau \in [0,T]$.
Next, define 
\begin{equation}\label{eq:def:Psi_m}
\Psi_\rmm \triangleq \{(x,\theta) \in \BBR^n \times \BBR^d \colon k(\theta) \geq 0, \psi_\rmm(x,\theta) \geq 0\},
\end{equation}
which is the set of $(x,\theta)$ such that $\phi(\tau;x,\theta) \in \SC_\rms$ and $u_\rmp(\tau;\theta) \in \SU$ for all $\tau \in [0,T]$.
Hence, $\Psi_\rmm \subseteq \SC_\rms \times \Theta$.

The next result demonstrates that $\psi_\rmm$ is locally Lipschitz and directionally differentiable, which implies that it satisfies the conditions in Definition~\ref{def:pfcbf} to be a candidate P-CBF.
The result also provides an expression for the directional derivative of $\psi_\rmm$, which depends on the following set
\begin{equation*}
\ST(x,\theta) \triangleq \operatorname*{argmin}_{\tau \in [0,T]} h_\rms(\phi(\tau; x, \theta)).
\end{equation*}
The proof is in the appendix.

\begin{proposition}\label{prop:directional_derivative_danskin}
\rm
Consider $\psi_\rmm$ given by \eqref{eq:def:psi_m}, where $h_\rms$ is continuously differentiable on $\BBR^n$.
Then, the following hold:
\begin{enumerate}[leftmargin=0.9cm]

\item\label{prop:directional_derivative_danskin.b}
$\psi_\rmm$ is locally Lipschitz on $\BBR^n \times \BBR^d$.

\item\label{prop:directional_derivative_danskin.c}
$\psi_\rmm$ is directionally differentiable on $\BBR^n \times \BBR^d$, and 
\begin{align}\label{eq:dini_psi_formula_danskin}
D_\nu\psi_\rmm(x,\theta) &= \min_{\tau \in \ST(x,\theta)} h_\rms^\prime(\phi(\tau; x, \theta)) \nn \\
&\quad \times \begin{bmatrix} \frac{\partial \phi}{\partial x}(\tau; x, \theta) & \frac{\partial \phi}{\partial \theta}(\tau; x, \theta) \end{bmatrix} \nu.
\end{align}

\end{enumerate}
\end{proposition}

The next result shows that the directional derivative \eqref{eq:dini_psi_formula_danskin} of $\psi_\rmm$ along the trajectories of \eqref{eq:affine_control} and \eqref{eq:theta_dyn} satisfies a given lower bound if and only if a family of related lower bounds are satisfied, where each is affine in the control variables $(\hat u, \hat \omega)$. 
This result is an immediate consequence of part~\ref{prop:directional_derivative_danskin.c} of Proposition~\ref{prop:directional_derivative_danskin}.

\begin{proposition}
\label{prop:dcbf_iff}
\rm
Let $c \in \BBR$, and let $(x,\theta,\hat u,\hat\omega) \in \BBR^n \times \BBR^d \times \SU \times \Omega$.
Then, 
\begin{equation*}
D_{\left[\begin{smallmatrix} F(x,\hat{u}) \\ \hat{\omega} \end{smallmatrix}\right]} \psi_\rmm(x,\theta) \ge c 
\end{equation*}
if and only if for all $\tau \in \ST(x,\theta)$,
\begin{align}
&h_\rms^\prime(\phi(\tau; x, \theta)) \begin{bmatrix} \frac{\partial \phi}{\partial x}(\tau; x, \theta) & \frac{\partial \phi}{\partial \theta}(\tau; x, \theta) \end{bmatrix} \nn\\
&\qquad \times \begin{bmatrix} f(x) + g(x) \hat u \\ \hat \omega \end{bmatrix} \ge c. \label{eq:Dpsi_m_affine}
\end{align}
\end{proposition}

\begin{remark}
\rm
Proposition~\ref{prop:dcbf_iff} implies that the control variables $(\hat u, \hat \omega)$ satisfy the directional derivative constraint 
\begin{equation*}
D_{\left[\begin{smallmatrix} F(x,\hat{u}) \\ \hat{\omega} \end{smallmatrix}\right]} \psi_\rmm(x,\theta) +\alpha(\psi_\rmm(x,\theta)) \ge 0
\end{equation*}
if and only if $(\hat u, \hat \omega)$ satisfy the family of affine constraints \eqref{eq:Dpsi_m_affine} with $c=-\alpha(\psi_\rmm(x,\theta))$. 
Similar to standard CBFs, these affine constraints are useful for control synthesis. 
\end{remark}

\subsection{TERMINAL-PREDICTION-TIME P-CBF}

This subsection presents a candidate P-CBF for the situation in which it is desirable for the predicted flow $\phi$ to be in a desired set at the terminal prediction time (i.e., $\tau=T$).
Specifically, let $h_\rmb \colon \BBR^n \to \BBR$ be continuously differentiable, and define
\begin{equation}\label{eq:backup_set}
\SC_\rmb \triangleq \left\{ x \in \BBR^n \colon h_\rmb(x) \geq 0 \right\}.
\end{equation}
Then, consider the candidate P-CBF
\begin{equation}\label{eq:def:psi_t}
\psi_\rmt(x,\theta) \triangleq h_\rmb(\phi(T ; x,\theta)),
\end{equation}
and associate set
\begin{equation}\label{eq:def:Psi_t}
\Psi_\rmt \triangleq \{(x,\theta) \in \BBR^n \times \BBR^d \colon k(\theta) \geq 0, \psi_\rmt(x,\theta) \geq 0\}.
\end{equation}
Note that $\Psi_\rmt$ is the set of $(x,\theta)$ such that $\phi(T;x,\theta) \in \SC_\rmb$ and $u_\rmp(\tau;\theta) \in \SU$ for all $\tau \in [0,T]$.

The next result demonstrates that $\psi_\rmt$ is continuously differentiable.
Thus, the directional derivative equals the Lie derivative, and the Lie derivative along the trajectories of \cref{eq:affine_control,eq:theta_dyn} is affine in the control variables $(\hat u, \hat \omega)$.

\begin{proposition}
\label{prop:psi_t_smooth}
\rm
Consider $\psi_\rmt$ given by \eqref{eq:def:psi_t}, where $h_\rmb$ is continuously differentiable on $\BBR^n$.
Then, $\psi_\rmt$ is continuously differentiable on $\BBR^n \times \BBR^d$, and
\begin{equation}\label{eq:psi_t_derivative}
D_{\left[\begin{smallmatrix} F(x,\hat{u}) \\ \hat{\omega} \end{smallmatrix}\right]} \psi_\rmt(x,\theta) = L_{\left[\begin{smallmatrix}f(x)+g(x)\hat{u} \\ \hat{\omega} \end{smallmatrix}\right]} \psi_\rmt(x,\theta).
\end{equation}
\end{proposition}

\begin{proof}
\indent Since $h_\rmb$ is continuously differentiable and $\phi(T;x,\theta)$ is continuously differentiable on $\BBR^n \times \BBR^d$, it follows that $\psi_\rmt(x,\theta) = h_\rmb(\phi(T;x,\theta))$ is continuously differentiable on $\BBR^n \times \BBR^d$.
Thus, $D_\nu \psi_\rmt(x,\theta) = L_\nu \psi_\rmt(x,\theta)$.
\end{proof}

\begin{remark}\rm
Equations \eqref{eq:def:psi_m} and \eqref{eq:def:psi_t} are two useful candidate P-CBFs used in this article. 
However, mapping the predicted flow $\phi$ through other functionals can yield other potentially useful candidate P-CBFs.
For example, consider the candidate P-CBF obtained by integrating a function of the flow over time, specifically, $\psi_{\rm int}(x,\theta) = a + \int_0^T b(\phi(\tau; x, \theta)) \, \rmd\tau$, where $a \in \BBR$ and $b \colon \BBR^n \to \BBR$ is continuously differentiable.
This candidate P-CBF can be used to capture the energy of the predicted flow. 
\end{remark}

\section{PROBLEM FORMULATION}\label{sec:problem}

For the remainder of this article, we consider the problem of designing admissible feedback controls $(u,\omega)$ that minimize an integral cost of the predicted flow $\phi$ over the prediction horizon such that the predicted flow $\phi(\cdot;x(t),\theta(t))$ and actual state $x(t)$ are in a prescribed \textit{safe set} $\SC_\rms$ for all time $t\ge 0$. 
This section precisely formulates this problem and presents solutions in 2 special circumstances: $\psi_\rmm$ is a P-CBF, or $(\psi_\rmm,\psi_\rmt)$ is a P-CBF pair. 
Although these special circumstances are useful for understanding, the associated P-CBF assumptions can be difficult to satisfy and/or verify. 
This motivates the idea of the planning-time shift that is introduced in Section~\ref{sec:control_forward_invariant} to avoid a P-CBF assumption.

The safe set $\SC_\rms$ is given by \eqref{eq:safe_set}, where $h_\rms$ is continuously differentiable and known.
Notably, $h_\rms$ is not assumed to be a CBF for \eqref{eq:affine_control} on $\SC_\rms$.
Similarly, $\psi_\rmm$ given by \eqref{eq:def:psi_m} is not assumed to be a P-CBF for \cref{eq:affine_control,eq:theta_dyn} on $\Psi_\rmm$, which is given by \eqref{eq:def:Psi_m}. 
Thus, it is not necessarily possible to make $\SC_\rms$ or $\Psi_\rmm$ forward invariant. 
We assume the admissible control sets $\SU$ and $\Omega$ are convex, and $0 \in \Omega$.

Next, consider the cost $J \colon \BBR^{n} \times \BBR^d \to \BBR$ given by
\begin{equation} \label{eq:def:J}
J(x,\theta) \triangleq W \left ( \phi(T; x,\theta) \right )
 + \int_{0}^{T} R\left(\phi(\tau; x,\theta),u_\rmp(\tau;\theta) \right) \, {\rm d} \tau,
\end{equation}
where $W \colon \BBR^n \to \BBR$ and $R \colon \BBR^n \times \BBR^m \to \BBR$ are continuously differentiable.
The objective is to minimize $J$ while ensuring the predicted flow is in $\SC_\rms$ and the control and control-plan parameters are admissible.
The objective is formalized as follows.

\begin{problem}\label{prob:main}\rm
Design feedback controls for $u(t) \in \SU$ and $\omega(t) \in \Omega$ such that for each time $t\ge 0$, the receding horizon cost $J(x(t),\theta(t))$ is optimized subject to the constraints:
\begin{enumerate}[leftmargin=0.9cm]
\renewcommand{\labelenumi}{(C\arabic{enumi})}
\renewcommand{\theenumi}{(C\arabic{enumi})}

\item\label{con:C1} For all $(t, \tau) \in [0,\infty) \times [0,T]$, $\phi(\tau; x(t),\theta(t)) \in \SC_\rms$.

\item\label{con:C2} For all $t \ge 0$, $\theta(t) \in \Theta$.

\end{enumerate}
\end{problem}

Constraints \ref{con:C1} and \ref{con:C2} can be re-framed in terms of forward invariance. 
Note that \ref{con:C1} and \ref{con:C2} are satisfied if and only if for all $t\ge 0$, $(x(t),\theta(t)) \in \Psi_\rmm$. 
Thus, \ref{con:C1} and \ref{con:C2} are satisfied by designing $u(t) \in \SU$ and $\omega(t) \in \Omega$ that make a subset of $\Psi_\rmm$ forward invariant.

Since $h_\rms$ is not assumed to be a CBF, and $\psi_\rmm$ is not assumed to be a P-CBF, it may not be possible to satisfy \ref{con:C1} and \ref{con:C2}. 
To ensure the problem is well posed, we assume there exists a subset of $\SC_\rms$ that can be made forward invariant with respect to \eqref{eq:affine_control}. 
Specifically, consider a \textit{backup safe set} $\SC_\rmb \subset \SC_\rms$, which is given by \eqref{eq:backup_set}, where $h_\rmb$ is continuously differentiable and known. 
We make the following assumption.

\begin{assumption}\label{ass:backup_set}\rm
There exists a known extended class-$\SK$ function $\alpha_\rmb \colon \BBR \to \BBR$ such that for all $x \in \SC_\rmb$, 
\begin{equation*}
K_\rmb(x) \triangleq \{ \hat u \in \SU \colon L_fh_\rmb(x) + L_gh_\rmb(x) \hat u + \alpha_\rmb(h_\rmb(x)) > 0 \}
\end{equation*}
is nonempty.
\end{assumption}

Assumption~\ref{ass:backup_set} implies that $h_\rmb$ is a CBF for \eqref{eq:affine_control} on $\SC_\rmb$.
The assumption also requires knowledge of the function $\alpha_\rmb$. 
In certain cases (i.e., polynomial dynamics and semialgebraic constraints), sum-of-squares methods can be used to synthesize this function \cite{pond2025verification}.
More generally, if $K_\rmb(x)$ is nonempty for all $x \in \bd\,\SC_\rmb$, then there exists an extended class-$\SK$ function $\alpha_\rmb$ such that $K_\rmb(x)$ is nonempty on the interior of $\SC_\rmb$. 
In this case, Assumption~\ref{ass:backup_set} can be satisfied by selecting $\alpha_\rmb$ with a sufficiently large derivative.
%
%For systems with polynomial dynamics and semialgebraic constraint sets, sum-of-squares methods can be used to verify or synthesize such functions \cite{pond2025verification}.
%More generally, note that on $\bd\,\SC_\rmb$, where $h_\rmb(x) = 0$, the condition in Assumption~\ref{ass:backup_set} is independent of $\alpha_\rmb$ and reduces to
%\begin{equation*}
%\sup_{\hat u \in \SU} \; L_f h_\rmb(x) + L_g h_\rmb(x) \hat u > 0.
%\end{equation*}
%Thus, if this condition holds on $\bd\,\SC_\rmb$, then $\alpha_\rmb$ can be selected to satisfy %Assumption~\ref{ass:backup_set} on the interior of $\SC_\rmb$, where $h_\rmb(x) > 0$.}

Assumption~\ref{ass:backup_set} also implies that $h_\rmb$ has relative degree one with respect to \eqref{eq:affine_control}; however, this can be directly extended to higher relative degree using standard higher-order CBF approaches \cite{xiao2021high}.

Under Assumption~\ref{ass:backup_set}, $\SC_\rmb$ can be made forward invariant with respect to \eqref{eq:affine_control}.
However, $\SC_\rmb$ may be small relative to $\SC_\rms$.
Thus, it is not desirable for the predicted flow $\phi(\tau;x(t),\theta(t))$ to be in $\SC_\rmb$ for all prediction times $\tau \in [0,T]$ or for all real time $t \ge0$.
Such a restriction on $\phi$ can lead to large $J$, that is, poor performance.
Instead, Assumption~\ref{ass:backup_set} is used to introduce a terminal-prediction-time condition on the predicted flow.
Specifically, we consider the condition that for all time $t\ge 0$, the predicted flow $\phi(\cdot;x(t),\theta(t))$ at terminal prediction time $T$ is in $\SC_\rmb$.
This condition can be re-framed in terms of forward invariance by considering $\psi_\rmt$ and $\Psi_\rmt$ given by \eqref{eq:def:psi_t} and \eqref{eq:def:Psi_t}. 
Specifically, for all $t\ge 0$, $\phi(T;x(t),\theta(t)) \in \SC_\rmb$ and $\theta(t) \in \Theta$ if and only if for all $t\ge 0$, $(x(t),\theta(t)) \in \Psi_\rmt$. 
Thus, the terminal condition can be satisfied by designing controls $u(t) \in \SU$ and $\omega(t) \in \Omega$ that make a subset of $\Psi_\rmt$ forward invariant.

Since $J$ can be nonlinear and nonconvex, Problem~\ref{prob:main} cannot necessarily be solved with a convex optimization. 
However, the time derivative of $J$ along the trajectories of \eqref{eq:affine_control} and \eqref{eq:theta_dyn} is
\begin{equation}
\frac{\rmd J}{\rmd t} = \frac{\partial J(x,\theta)}{\partial x} \left [ f(x) + g(x) u \right ] + \frac{\partial J(x,\theta)}{\partial \theta} \omega.
\end{equation}
To make the time derivative of $J$ small, consider the quadratic cost
\begin{align}\label{eq:SJ}
    \SJ(\hat u,\hat \omega; x,\theta) &\triangleq \frac{\partial J(x,\theta)}{\partial x} g(x) \hat u + \frac{\partial J(x,\theta)}{\partial \theta} \hat{\omega} + \hat{\omega}^\top Q_\omega \hat{\omega} \nn\\
    &\quad + [\hat{u} - u_\rmp(0; \theta)]^\top Q_u [\hat{u} - u_\rmp(0; \theta)],
\end{align}
where $Q_u \in \BBR^{m \times m}$ and $Q_\omega \in \BBR^{d \times d}$ are positive definite.
The first 2 terms of \eqref{eq:SJ} make $\frac{\rmd J}{\rmd t}$ small, whereas the other 2 terms provide regularization that make \eqref{eq:SJ} strictly convex.
Specifically, $\hat{\omega}^\top Q_\omega \hat{\omega}$ limits the rate of change of the control plan parameter $\theta$ and $[\hat{u} - u_\rmp(0; \theta)]^\top Q_u [\hat{u} - u_\rmp(0; \theta)]$ limits deviation of the executed control from the control plan.
Hence, minimizing the quadratic cost $\SJ$ subject to \ref{con:C1} and \ref{con:C2} yields gradient flow that aims to decrease $J$ along the trajectories of \eqref{eq:affine_control} and \eqref{eq:theta_dyn}.
Thus, we address Problem~\ref{prob:main} by solving the following problem.

\begin{problem}\label{prob:gradient_flow}\rm
Design feedback controls for $u(t) \in \SU$ and $\omega(t) \in \Omega$ such that for each time $t \ge 0$, the quadratic cost $\SJ(\hat u, \hat \omega; x(t), \theta(t))$ is minimized subject to \ref{con:C1} and \ref{con:C2}.
\end{problem}

The following subsections present solutions to Problem~\ref{prob:gradient_flow} in special circumstances: \textit{A.} $\psi_\rmm$ is a P-CBF, and \textit{B.} $(\psi_\rmm,\psi_\rmt)$ is a P-CBF pair.
These special circumstances can be difficult to satisfy and/or verify, which motivates the remainder of this article, where we present a solution to Problem~\ref{prob:gradient_flow} without these assumptions.

\subsection{SOLUTION IF $\psi_\rmm$ IS A P-CBF}

This subsection addresses the case where $\psi_\rmm$ is a P-CBF on $\Psi_\rmm$ for \eqref{eq:affine_control} and \eqref{eq:theta_dyn}.
In this case, Definition~\ref{def:pfcbf} implies that there exist extended class-$\SK$ functions $\alpha,\beta$ such that for all $(x,\theta) \in \Psi_\rmm$,
\begin{align}\label{eq:def:Km}
K_\rmm(x,\theta) &\triangleq \{(\hat{u}, \hat{\omega}) \in \SU \times \Omega \colon k^\prime(\theta)\,\hat{\omega} + \beta(k(\theta)) \geq 0, \nn\\
&\quad D_{\left[\begin{smallmatrix} F(x,\hat{u}) \\ \hat{\omega} \end{smallmatrix}\right]} \psi_\rmm(x,\theta) + \alpha(\psi_\rmm(x,\theta)) \geq 0 \}
\end{align}
is nonempty.
The next result shows that $\SJ$ has a unique global minimizer over the set $K_\rmm(x,\theta)$.

\begin{proposition}\label{prop:K_convex}
\rm
Assume $\SU$ and $\Omega$ are convex, and $\psi_\rmm$ is a P-CBF on $\Psi_\rmm$ for \eqref{eq:affine_control} and \eqref{eq:theta_dyn}.
Then, for all $(x,\theta) \in \Psi_\rmm$, $K_\rmm(x,\theta)$ is convex, and $\SJ(\hat u,\hat \omega; x,\theta)$ has a unique global minimizer over $K_\rmm(x,\theta)$.
\end{proposition}

\begin{proof}
\indent Since $F(x,\hat{u}) = f(x) + g(x)\hat{u}$ is affine in $\hat{u}$, it follows from \eqref{eq:def:Km} and Proposition~\ref{prop:dcbf_iff} that $K_\rmm(x,\theta)$ is the intersection of the convex sets $\SU$ and $\Omega$ with a family of affine half-spaces, and thus, it is convex.

Since $K_\rmm(x,\theta)$ is nonempty and convex, and $\SJ$ is strictly convex in $(\hat{u},\hat{\omega})$, the minimizer exists and is unique.
\end{proof}

The next result solves Problem~\ref{prob:gradient_flow} where $\psi_\rmm$ is a P-CBF.

\begin{corollary}\label{cor:pcbf_solution}\rm
Assume $\SU$ and $\Omega$ are convex, and $\psi_\rmm$ is a P-CBF on $\Psi_\rmm$ for \eqref{eq:affine_control} and \eqref{eq:theta_dyn}.
For all $(x,\theta) \in \Psi_\rmm$, define
\begin{equation}\label{eq:pcbf_solution}
\left ( u_*(x,\theta),\omega_*(x,\theta) \right )  \triangleq \operatorname*{argmin}_{(\hat{u},\hat{\omega}) \in K_\rmm(x,\theta)}     \SJ(\hat u,\hat \omega; x,\theta).
\end{equation}
Then, $\Psi_\rmm$ is forward invariant with respect to \eqref{eq:affine_control} and \eqref{eq:theta_dyn} with $(u,\omega) = (u_* , \omega_*)$.
Furthermore, \ref{con:C1} and \ref{con:C2} are satisfied.
\end{corollary}

\begin{proof}
\indent Since $\psi_\rmm$ is a P-CBF and $(u_*(x,\theta), \omega_*(x,\theta)) \in K_\rmm(x,\theta)$, Theorem~\ref{thm:pfcbf_forward_invariance} implies that $\Psi_\rmm$ is forward invariant with respect to \eqref{eq:affine_control} and \eqref{eq:theta_dyn} with $(u,\omega) = (u_*, \omega_*)$.

Since $(x(t),\theta(t)) \in \Psi_\rmm$ implies $\psi_\rmm(x(t),\theta(t)) \geq 0$ and $k(\theta(t)) \geq 0$, it follows from \cref{eq:def:psi_m,eq:theta_set} that \ref{con:C1} and \ref{con:C2} are satisfied.
\end{proof}

\begin{remark}
\label{remark:PCBF-QP}
\rm
If $\SU$ and $\Omega$ are convex polyhedra, then Proposition~\ref{prop:dcbf_iff} implies that all constraints that define $K_\rmm(x,\theta)$ are affine.
In this case, the optimal control \eqref{eq:pcbf_solution} can be obtained efficiently from a QP that depends on the predicted flow.
Section~\ref{sec:impl_remark} addresses implementation of these QPs.
\end{remark}

\begin{remark}
\label{remark:PCBF-QP.2}
\rm
Corollary~\ref{cor:pcbf_solution} requires that $\psi_\rmm$ is a P-CBF, which is generally difficult to satisfy and verify because it requires that for each $(x,\theta) \in \Psi_\rmm$, there exist $(\hat{u},\hat{\omega}) \in \SU\times\Omega$ that satisfies both constraints in $K_\rmm(x,\theta)$.
In general, control input constraints (i.e., $\SU$ and $\Omega$) can lead to points where $K_\rmm(x,\theta)$ is empty.
Even if $\psi_\rmm$ is a P-CBF, it may be difficult to determine the class-$\SK$ functions $\alpha,\beta$ such that $K_\rmm(x,\theta)$ is nonempty for all $(x,\theta) \in \Psi_\rmm$.
\end{remark}

\subsection{SOLUTION IF $(\psi_\rmm,\psi_\rmt)$ IS A P-CBF PAIR}

This subsection addresses the case where $(\psi_\rmm,\psi_\rmt)$ is a P-CBF pair on $\Psi_{\rmm \rmt} \triangleq \Psi_\rmm \cap \Psi_\rmt$, which is a subset of $\Psi_\rmm$ created by imposing the condition that the predicted flow $\phi$ is in the backup safe set at the terminal prediction time. 
Since $(\psi_\rmm,\psi_\rmt)$ is a P-CBF pair, Definition~\ref{def:pfcbf} implies that there exist extended class-$\SK$ functions $\alpha_1,\alpha_2,\beta$ such that for all $(x,\theta) \in \Psi_{\rmm \rmt}$,
\begin{align}\label{eq:def:Kmt}
K_{\rmm \rmt}(x,\theta) &\triangleq \{(\hat{u}, \hat{\omega}) \in \SU \times \Omega \colon k^\prime(\theta)\,\hat{\omega} + \beta(k(\theta)) \geq 0, \nn\\
&\qquad D_{\left[\begin{smallmatrix} F(x,\hat{u}) \\ \hat{\omega} \end{smallmatrix}\right]} \psi_\rmm(x,\theta) + \alpha_1(\psi_\rmm(x,\theta)) \geq 0, \nn\\
&\qquad L_{\left[\begin{smallmatrix} f(x)+g(x)\hat{u} \\ \hat{\omega} \end{smallmatrix}\right]} \psi_\rmt(x,\theta) + \alpha_2(\psi_\rmt(x,\theta)) \geq 0 \}
\end{align}
is nonempty.
The next result shows that $\SJ$ has a unique global minimizer over $K_{\rmm \rmt}(x,\theta)$.
The proof is similar to that of Proposition~\ref{prop:K_convex}.

\begin{proposition}\label{prop:K_convex_mt}
\rm
Assume $\SU$ and $\Omega$ are convex, and $(\psi_\rmm,\psi_\rmt)$ is a P-CBF pair on $\Psi_{\rmm \rmt}$ for \eqref{eq:affine_control} and \eqref{eq:theta_dyn}.
Then, for all $(x,\theta) \in \Psi_{\rmm \rmt}$, $K_{\rmm \rmt}(x,\theta)$ is convex, and $\SJ(\hat u,\hat \omega; x,\theta)$ has a unique global minimizer over $K_{\rmm \rmt}(x,\theta)$.
\end{proposition}

The following result solves Problem~\ref{prob:gradient_flow} where $(\psi_\rmm,\psi_\rmt)$ is a P-CBF pair.
The proof is similar to that of Corollary~\ref{cor:pcbf_solution}.

\begin{corollary}
\label{cor:pcbf_pair_solution}
\rm
Assume $\SU$ and $\Omega$ are convex, and $(\psi_\rmm,\psi_\rmt)$ is a P-CBF pair on $\Psi_{\rmm \rmt}$ for \eqref{eq:affine_control} and \eqref{eq:theta_dyn}.
For all $(x,\theta) \in \Psi_{\rmm \rmt}$, define
\begin{equation}\label{eq:pcbf_pair_solution}
\left ( u_*(x,\theta),\omega_*(x,\theta) \right )  \triangleq \operatorname*{argmin}_{(\hat{u},\hat{\omega}) \in K_{\rmm \rmt}(x,\theta)}     \SJ(\hat u,\hat \omega; x,\theta).
\end{equation}
Then, $\Psi_{\rmm \rmt}$ is forward invariant with respect to \eqref{eq:affine_control} and \eqref{eq:theta_dyn} with $(u,\omega) = (u_* , \omega_*)$.
Furthermore, \ref{con:C1} and \ref{con:C2} are satisfied.
\end{corollary}

\begin{remark}\rm
Similar to Remark~\ref{remark:PCBF-QP}, if $\SU$ and $\Omega$ are convex polyhedra, then Proposition~\ref{prop:dcbf_iff} implies that all constraints that define $K_{\rmm \rmt}(x,\theta)$ are affine.
In this case, the optimal control \eqref{eq:pcbf_pair_solution} is the solution to a QP.
\end{remark}

\begin{remark}\rm
The optimal control \eqref{eq:pcbf_pair_solution} in this subsection makes $\Psi_{\rmm \rmt}$ forward invariant in the case where $(\psi_\rmm,\psi_\rmt)$ is a P-CBF pair.
It is worth noting that $\Psi_{\rmm \rmt} \subset \Psi_{\rmm}$ because \eqref{eq:pcbf_pair_solution} is concerned not only with keeping the predicted flow $\phi(\cdot; x(t),\theta(t))$ in the safe set $\SC_\rms$ but also with keeping the predicted flow $\phi(T; x(t),\theta(t))$ at the terminal prediction time $T$ in the backup safe set $\SC_\rmb$.
This terminal-prediction constraint focuses on achieving forward invariance for the subset $\Psi_{\rmm \rmt}$ because $\Psi_{\rmm}$ cannot generally be made forward invariant if $\psi_\rmm$ is not a P-CBF. 
Intuitively, it may seem more likely that $\Psi_{\rmm \rmt}$ can be made forward invariant than $\Psi_{\rmm}$.
However, the constraints in $K_{\rmm \rmt}$ are also more restrictive than those in $K_{\rmm}$.
Similar to Remark~\ref{remark:PCBF-QP.2}, it is generally difficult to satisfy and/or verify the condition that $(\psi_\rmm,\psi_\rmt)$ is a P-CBF pair.
Even if $(\psi_\rmm,\psi_\rmt)$ is a P-CBF pair, it may be difficult to determine the class-$\SK$ functions $\alpha_1,\alpha_2,\beta$ such that $K_{\rmm \rmt}(x,\theta)$ is nonempty for all $(x,\theta) \in \Psi_{\rmm \rmt}$.
\end{remark}

The optimal controls in Corollaries~\ref{cor:pcbf_solution} and \ref{cor:pcbf_pair_solution} rely on assumptions that are difficult to satisfy and/or verify. 
The following section addresses these challenges by introducing a planning-time shift that guarantees feasibility of a convex optimization used to obtain optimal controls without requiring that $\psi_\rmm$ is a P-CBF or that $(\psi_\rmm,\psi_\rmt)$ is a P-CBF pair.
Moreover, this convex optimization is a QP in the case where $\SU$ and $\Omega$ are convex polyhedra.

\section{SAFE OPTIMAL FLOW CONTROL}
\label{sec:control_forward_invariant}

This section solves Problem~\ref{prob:gradient_flow} by introducing a planning-time shift that guarantees feasibility of the optimization that determines $(u,\omega)$. 
The key idea is to augment the state with a planning-time shift, which allows the prediction horizon to shrink as necessary to guarantee feasibility.

\subsection{PREDICTED FLOW WITH PLANNING TIME SHIFT}

We extend the definition \eqref{eq:flow_def} for the predicted flow $\phi(\cdot;x,\theta)$ to allow for a time shift in the control plan $u_\rmp(\cdot;\theta)$.
Let $\gamma \in [0,T]$ be the \textit{planning-time shift}, and let the predicted flow $\varphi(\cdot;x,\theta,\gamma) \colon [\gamma, T] \to \BBR^n$ satisfy
\begin{equation}\label{eq:flow_def_shifted}
\varphi(\tau; x, \theta, \gamma) =
x + \int_{\gamma}^{\tau} F \!\left( \varphi(\sigma; x, \theta, \gamma),\, u_\rmp(\sigma; \theta)\right) \, {\rm d} \sigma.
\end{equation}
For $\gamma=0$, $\varphi(\tau; x, \theta, 0)$ reduces to \eqref{eq:flow_def}.
In other words, $\varphi(\tau; x, \theta, 0)$ is the solution to
\eqref{eq:affine_control} at planning time $\tau$, where $x$ is the initial condition and $u(t) = u_\rmp(t; \theta)$.
For $\gamma \in (0,T]$, $\varphi(\tau; x, \theta, \gamma)$ is the solution to \eqref{eq:affine_control} at planning time $\tau-\gamma$, where $x$ is the initial condition and $u(t)=u_\rmp(t+\gamma;\theta)$, which implies that the control plan is shifted by $\gamma$.
Differentiating \eqref{eq:flow_def_shifted} with respect to $\tau$ yields
\begin{equation*}%\label{eq:predicted_dynamics}
\frac{\partial \varphi}{\partial \tau}(\tau; x, \theta, \gamma)
= F\big(\varphi(\tau; x, \theta, \gamma), u_\rmp(\tau; \theta)\big),
\end{equation*}
which is the evolution of the predicted flow $\varphi$ given $(x,\theta,\gamma)$.

The planning-time shift $\gamma$ is the continuous-time analogue to shrinking the horizon in discrete-time receding-horizon control.
In this work, we use the planning-time shift $\gamma$ to guarantee feasibility of an optimization that is used to generate the control and update the control plan. 
This can be viewed as the continuous-time analogue to the recursive-feasibility approach in discrete-time MPC \cite{mayne2000}. 
As $\gamma$ increases, the remaining prediction window $[\gamma, T]$ shrinks.

Similar to the approach used with $\theta$, we introduce dynamics to influence the time evolution of $\gamma$.
Specifically, let $\gamma \colon [0,\infty) \to \BBR$ be the solution to 
\begin{equation}\label{eq:gamma_dyn}
\dot{\gamma}(t) = z(t),
\end{equation}
where $\gamma(0) = \gamma_0 \in [0,T]$, $z \colon [0,\infty) \to \SZ$ is the control input to the integrator, $\SZ \subseteq \BBR$ is a closed interval, and $1\in \SZ$.

For notational convenience, we write \eqref{eq:affine_control}, \eqref{eq:theta_dyn}, and \eqref{eq:gamma_dyn} as
\begin{equation}\label{eq:augmented}
\dot{\bar{x}}(t) = \bar{f}(\bar{x}) + \bar{g}(\bar{x}) \bar{u}(t),
\end{equation}
where
 \begin{gather}
 \bar x \triangleq \begin{bmatrix} x \\ \theta\\ \gamma \end{bmatrix}, \quad
\bar x_0 \triangleq \begin{bmatrix} x_0 \\ \theta_0 \\ \gamma_0 \end{bmatrix},\quad
 \bar{u} \triangleq \begin{bmatrix} u \\ \omega \\ z \end{bmatrix}, \label{eq:augmented.2}\\
                 \bar{f}( \bar{x}) \triangleq \begin{bmatrix}
                       f(x)\\
                       0_{(d+1) \times 1}
                 \end{bmatrix},\quad
\bar{g}(\bar{x}) \triangleq \begin{bmatrix}
                   g(x) & 0_{n \times (d+1)} \\
                   0_{(d+1) \times m} & I_{d+1}
           \end{bmatrix},\label{eq:augmented.3}
\end{gather}
where $\bar n \triangleq n+d+1$ and $\bar m \triangleq m+d+1$.

\subsection{A CONTROL FORWARD INVARIANT SET}

Consider predicted-flow barrier functions
\begin{gather}
\bar\psi_\rmm(\bar{x}) \triangleq \min\limits_{\tau \in [\gamma, T]} h_{\rms}(\varphi(\tau ; \bar{x})), \label{eq:psi_min}\\
\bar \psi_\rmt(\bar{x}) \triangleq h_\rmb(\varphi(T ; \bar{x})), \label{eq:psi_b}
\end{gather}
which are analogous to \eqref{eq:def:psi_m} and \eqref{eq:def:psi_t} but incorporate the planning-time shift.
Notably, no relative degree assumption on $\bar\psi_\rmm$ or $\bar\psi_\rmt$ is required.
These functions can have arbitrary relative degree with respect to \eqref{eq:augmented}. 
Next, define 
\begin{align}\label{eq:def:barPsi}
\bar \Psi &\triangleq \{\bar{x} \in \BBR^{\bar n} \colon \bar\psi_\rmm(\bar{x}) \ge 0, \, \bar\psi_\rmt(\bar{x}) \ge 0, \nn\\
&\qquad k(\theta) \ge 0, \, \gamma \in [0,T] \},
\end{align}
which is the set of $\bar x$ such that the predicted flow satisfies $\varphi(\tau;\bar{x}) \in \SC_\rms$ for all prediction times $\tau \in [\gamma,T]$; the predicted flow satisfies the terminal condition $\varphi(T;\bar{x}) \in \SC_\rmb$; the control-plan parameter $\theta$ is in the admissible set $\Theta$; and the planning-time shift $\gamma$ is in the admissible set $[0,T]$.
Note that $\bar \Psi$ is analogous to $\Psi_{\rmm \rmt}$ but incorporates the planning-time shift.

The remainder of this subsection focuses on demonstrating that there exists a control $\bar u$ that makes $\bar \Psi$ forward invariant.
Consider the backup control $u_\rmb \colon \SC_\rmb \to \SU$ defined by
\begin{equation}\label{eq:backup_controller}
u_\rmb(x) \triangleq \operatorname*{argmin}_{\hat{u} \in K_\rmb(x)} \|\hat{u}\|^2,
\end{equation}
which exists and is unique because $\|\hat u \|^2$ is strictly convex, and $K_\rmb(x)$ is convex and nonempty from Assumption~\ref{ass:backup_set}.
The next result demonstrates that $u_\rmb$ makes $\SC_\rmb$ forward invariant.
The proof is in the appendix. 

\begin{proposition}\label{prop:backup_controller_exists}
\rm
Assume $\SU$ is convex and let Assumption~\ref{ass:backup_set} be satisfied.
Then, $\SC_\rmb$ is forward invariant with respect to \eqref{eq:affine_control} with $u=u_\rmb$.
Furthermore, if $\SU$ is compact, then $u_\rmb$ is continuous on $\SC_\rmb$.
\end{proposition}

Let $\bar{u}_{\rm fb} \colon \bar\Psi \to \SU \times \Omega \times \SZ$ be given by
\begin{equation}\label{eq:def:ubar_fb}
    \bar{u}_{\rm fb}(\bar x) \triangleq \begin{cases}
\bar{u}_\rmp(\bar x), & \gamma \in [0,T),\\
\bar{u}_\rmb(\bar x), & \gamma = T,
    \end{cases}
\end{equation}
where $\bar{u}_\rmp,\bar{u}_\rmb \colon \bar\Psi \to \SU \times \Omega \times \SZ$ are given by
\begin{equation}\label{eq:def:ubar_p_ubar_b}
\bar{u}_\rmp(\bar x) \triangleq \begin{bmatrix} u_\rmp(\gamma;\theta) \\ 0_{d \times 1} \\ 1 \end{bmatrix},\qquad
\bar{u}_\rmb(\bar x) \triangleq \begin{bmatrix} u_\rmb(x) \\ 0_{d \times 1} \\ 0 \end{bmatrix}.
\end{equation}
It follows from \cref{eq:theta_dyn,eq:gamma_dyn,eq:def:ubar_fb,eq:def:ubar_p_ubar_b} that the control $\bar u = \bar{u}_{\rm fb}$ results in constant control-plan parameters, that is, $\theta(t) = \theta_0$, and a planning-time shift that increases linearly with time until it reaches $T$, specifically, $\gamma(t) = \min \{ t+\gamma_0,T \}$.
Thus, $\bar{u}_{\rm fb}$ results in control $u$ that follows the control plan $u_\rmp$ for $t \in [0,T-\gamma_0)$ and switches to the backup control $u_\rmb$ at $t=T-\gamma_0$.
The next result shows that $\bar{u}_{\rm fb}$ makes $\bar\Psi$ forward invariant.

\begin{proposition}\label{prop:forward_inv}
\rm
Assume $\SU$ is convex and let Assumption~\ref{ass:backup_set} be satisfied.
Then, $\bar\Psi$ is forward invariant with respect to \cref{eq:augmented,eq:augmented.2,eq:augmented.3} with $\bar{u} = \bar{u}_{\rm fb}$.
Furthermore, for all $\bar x_0 \in \bar\Psi$, the following hold:
\begin{enumerate}[leftmargin=0.7cm]

\item \label{prop:forward_inv.3}
For all $t \in [T-\gamma_0,\infty)$, $x(t) \in \SC_\rmb$.

\item \label{prop:forward_inv.1}
If $\SU$ is compact, then $\bar u_{\rm fb}(\bar x(\cdot))$ is continuous on $[0,\infty) \setminus \{ T-\gamma_0 \}$.

\end{enumerate}
\end{proposition}

\begin{proof}
\indent Let $\bar{x}_0 \in \bar\Psi$. 
Since $\bar{u} = \bar{u}_{\rm fb}$, it follows from \cref{eq:theta_dyn,eq:gamma_dyn,eq:def:ubar_fb,eq:def:ubar_p_ubar_b} that $\theta(t) = \theta_0$, $\gamma(t) = \min \{ t+\gamma_0,T \}$, and 
\begin{equation}\label{eq:bar_u_fi_u}
    u(t) = \begin{cases}
u_\rmp(t+\gamma_0;\theta_0), & t \in [0,t_\rms),\\
u_\rmb(x(t)), & t \ge  t_\rms,
    \end{cases}
\end{equation}
where $t_\rms \triangleq T-\gamma_0$.
To show forward invariance of $\bar \Psi$, we consider the time intervals $[0,t_\rms]$ and $(t_\rms,\infty)$.

First, we show that for all $t \in [0,t_\rms]$, $\bar x(t) \in \bar \Psi$. 
Let $t_1 \in [0,t_\rms]$, and note that \cref{eq:affine_control,eq:flow_def_shifted,eq:bar_u_fi_u} imply that $x(t_1) = \varphi(t_1+\gamma_0;\bar{x}_0)$.
For notational convenience, define $x_1 \triangleq x(t_1)$, $\theta_1 \triangleq \theta(t_1)$, and $\gamma_1 \triangleq \gamma(t_1)$, and note that 
\begin{equation} \label{eq:state_equal_flow}
x_1 = \varphi(\gamma_1;\bar{x}_0), \quad \theta_1 = \theta_0, \quad \gamma_1 = t_1+\gamma_0. 
\end{equation}
It follows from \eqref{eq:flow_def_shifted} that for all $\tau \in [\gamma_0, T]$,
\begin{equation}
\varphi(\tau;\bar{x}_0) = x_0 + \int_{\gamma_0}^{\tau} F \left(\varphi(\sigma;\bar{x}_0),\, u_\rmp(\sigma;\theta_0)\right) \, {\rm d}\sigma,
\label{eq:flow_expand.a}
\end{equation}
which implies that for all $\tau \in [\gamma_1, T]$,
\begin{align}
\varphi(\tau;\bar{x}_0) 
&= x_0 + \int_{\gamma_0}^{\gamma_1} F \left(\varphi(\sigma;\bar{x}_0),\, u_\rmp(\sigma;\theta_0)\right) \, {\rm d}\sigma \nn\\
&\quad +\int_{\gamma_1}^{\tau} F \Big ( \varphi( \sigma; \varphi(\gamma_1;\bar x_0),\theta_0, \gamma_1 ) , \, u_\rmp(\sigma;\theta_0) \Big) \,{\rm d}\sigma\nn\\
&= \varphi(\gamma_1;\bar{x}_0) \nn\\
&\quad +\int_{\gamma_1}^{\tau} F \Big ( \varphi( \sigma; \varphi(\gamma_1;\bar x_0),\theta_0, \gamma_1 ) , \, u_\rmp(\sigma;\theta_0) \Big) \,{\rm d}\sigma.
\label{eq:flow_expand}
\end{align}
Substituting \eqref{eq:state_equal_flow} into \eqref{eq:flow_expand} and using \eqref{eq:flow_expand.a} implies that for all $\tau \in [\gamma_1, T]$,
\begin{align}
\varphi(\tau;\bar{x}_0) &= x_1 + \int_{\gamma_1}^{\tau}  F (\varphi\big(\sigma;\, x_1,\, \theta_1,\, \gamma_1 \big),\, u_\rmp(\sigma;\theta_0) ) \, {\rm d}\sigma\nn\\
&= x(t_1) + \int_{\gamma(t_1)}^{\tau}   F  (\varphi ( \sigma;\, \bar x(t_1)  ),\, u_\rmp(\sigma;\theta_0) ) \, {\rm d}\sigma\nn\\
&= \varphi(\tau;\bar{x}(t_1)).\label{eq:cocycle_integral}
\end{align}
Hence, \cref{eq:psi_min,eq:cocycle_integral} imply that
\begin{align*}
    \bar \psi_\rmm(\bar x(t_1)) &= \min\limits_{\tau \in [\gamma_1, T]} h_{\rms}(\varphi(\tau ; \bar{x}(t_1)))\\
    &= \min\limits_{\tau \in [\gamma_1, T]} h_{\rms}(\varphi(\tau ; \bar{x}_0))\\
    &\ge \min\limits_{\tau \in [\gamma_0, T]} h_{\rms}(\varphi(\tau ; \bar{x}_0))\\
    &= \bar \psi_\rmm(\bar x_0) \ge 0.
\end{align*}
Similarly, \cref{eq:psi_b,eq:cocycle_integral} imply that
\begin{align*}
\bar \psi_\rmt(\bar x(t_1)) = h_{\rmb}(\varphi( T ; \bar{x}(t_1))) = h_{\rmb}(\varphi(T ; \bar{x}_0)) = \bar \psi_\rmt(\bar x_0) \ge 0.
\end{align*}
Since, in addition, $k(\theta(t_1)) = k(\theta_0) \ge0$ and $\gamma(t_1) = \gamma_1 \in [0,T]$, it follows from \eqref{eq:def:barPsi} that $\bar x(t_1) \in \bar \Psi$.
Thus, for all $t \in [0, t_\rms]$, $\bar x(t) \in \bar \Psi$.

Next, we show that for all $t \in [t_\rms,\infty)$, $\bar x(t) \in \bar \Psi$. 
Since $\bar x(t_\rms) \in \bar \Psi$ and $\gamma(t_\rms) = T$, it follows from \eqref{eq:psi_b} that
$h_\rmb(x(t_\rms)) = \bar \psi_\rmt(\bar x(t_\rms)) \ge 0$, which implies $x(t_\rms) \in \SC_\rmb$.
Since $x(t_\rms) \in \SC_\rmb$ and $u(t) = u_\rmb(x(t))$, Proposition~\ref{prop:backup_controller_exists} implies that for all $t \in [t_\rms,\infty)$, $x(t) \in \SC_\rmb$, which confirms \ref{prop:forward_inv.3}.
Since, in addition, $\SC_\rmb \subset \SC_\rms$, and for all $t \in [t_\rms,\infty)$, $\theta(t)= \theta_0$ and $\gamma(t)= T$, it follows from \cref{eq:psi_min,eq:psi_b} that for all $t \in [t_\rms,\infty)$, $\bar\psi_\rmm(\bar{x}(t)) = h_\rms(x(t)) \ge 0$, and $\bar\psi_\rmt(\bar{x}(t)) = h_\rmb(x(t)) \ge 0$.
Thus, for all $t \in [t_\rms,\infty)$, $\bar x(t) \in \bar\Psi$, which confirms that $\bar \Psi$ is forward invariant. 

To prove \ref{prop:forward_inv.1}, Proposition~\ref{prop:backup_controller_exists} implies that $u_\rmb$ is continuous on $\SC_\rmb$.
Since, in addition, $u_\rmp(\,\cdot\,;\theta_0)$ is continuous on $[0,T]$, it follows from \eqref{eq:bar_u_fi_u} that $u$ is continuous on $[0,\infty) \setminus \{t_\rms\}$, which combined with \cref{eq:def:ubar_fb,eq:def:ubar_p_ubar_b} confirms \ref{prop:forward_inv.1}.
\end{proof}

Proposition~\ref{prop:forward_inv} demonstrates that $\bar\Psi$ is forward invariant with $\bar{u} = \bar{u}_{\rm fb}$.
This construction does not involve any explicit requirement on the relative degree of $\bar\psi_\rmm$ or $\bar\psi_\rmt$.

\subsection{CONTROLS THAT MAKE $\bar \Psi$ FORWARD INVARIANT}

Although $\bar{u}_{\rm fb}$ makes $\bar\Psi$ forward invariant, this control does not generally optimize the cost. 
Thus, this section provides a set of controls $\bar u$ that make $\bar\Psi$ forward invariant. 
In fact, this section shows that $\bar\Psi$ is made forward invariant by any control that satisfies P-CBF-like constraints, which are similar to \cref{eq:def:Kmt}.

First, we extend Proposition~\ref{prop:directional_derivative_danskin} to address the planning-time shift.
The result demonstrates that $\bar\psi_\rmm$ is locally Lipschitz and directionally differentiable, and provides an expression for the directional derivative of $\bar\psi_\rmm$, which depends on the following set
\begin{equation*}
\overline{\ST}(\bar{x}) \triangleq \operatorname*{argmin}_{\tau \in [\gamma, T]} h_\rms(\varphi(\tau; \bar{x})).
\end{equation*}
The proof is in the appendix.

\begin{proposition}\label{prop:psi_lipschitz}\rm
Consider $\bar \psi_\rmm$ given by \eqref{eq:psi_min}, where $h_\rms$ is continuously differentiable on $\BBR^n$.
Then, the following hold:
\begin{enumerate}[leftmargin=0.9cm]

\item\label{prop:psi_lipschitz_c}
$\bar\psi_\rmm$ is locally Lipschitz on $\BBR^{n} \times \BBR^d \times [0,T]$.

\item\label{prop:psi_lipschitz_d}
\label{prop:directional_derivative}
$\bar\psi_\rmm$ is directionally differentiable on $\{\bar{x} \in \bar\Psi \colon \gamma \neq T\}$, and for all $\bar{x} \in \{\bar{x} \in \bar\Psi \colon \gamma \neq T\}$, 
\begin{align}\label{eq:dini_psi_formula}
D_\nu\bar\psi_\rmm(\bar{x}) &= \min_{\tau \in \overline{\ST}(\bar{x})} h_\rms'(\varphi(\tau; \bar{x})) \left[ \frac{\partial \varphi}{\partial \bar{x}}(\tau; \bar{x}) \right. \nn \\
&\qquad + \left. \mathbf{1}_{\{\tau = \gamma\}} \frac{\partial \varphi}{\partial \tau}(\tau; \bar{x}) \frac{\partial \gamma}{\partial \bar{x}} \right]\nu,
\end{align}
where $\mathbf{1}_{\{\tau = \gamma\}}$ is the indicator function.
\end{enumerate}
\end{proposition}

\begin{remark}\rm
The directional derivative \eqref{eq:dini_psi_formula} differs from \eqref{eq:dini_psi_formula_danskin} because the planning-time shift $\gamma$ appears in the left end point of the feasible set $[\gamma, T]$ for the minimization \eqref{eq:psi_min}. 
In contrast, the minimization \eqref{eq:def:psi_m} without the planning-time shift is over the constant feasible set $[0,T]$. 
The inclusion of $\gamma$ results in the extra term $\mathbf{1}_{\{\tau = \gamma\}} \frac{\partial \varphi}{\partial \tau} \frac{\partial \gamma}{\partial \bar{x}}$, which impacts the directional derivative if and only if the minimizer is at the left endpoint $\gamma$ and accounts for the effect of moving the feasible set $[\gamma, T]$ on $\bar\psi_\rmm$.
\end{remark}

The next result extends Proposition~\ref{prop:dcbf_iff} to address the planning-time shift.
This result is an immediate consequence of part~\ref{prop:psi_lipschitz_d} of Proposition~\ref{prop:psi_lipschitz}.

\begin{proposition}\label{prop:dcbf_iff_augmented}\rm
Let $\bar{x} \in \bar\Psi$ with $\gamma \in [0,T)$, $\hat{\bar{u}} \in \SU \times \Omega \times \SZ$, and $c \in \BBR$.
Then,
\begin{equation*}
D_{\bar{f}(\bar{x}) + \bar{g}(\bar{x})\hat{\bar{u}}} \bar\psi_\rmm(\bar{x}) \ge c
\end{equation*}
if and only if for all $\tau \in \overline{\ST}(\bar{x})$,
\begin{align}\label{eq:affine_constraints_new}
&h_\rms'(\varphi(\tau; \bar{x})) \left[ \frac{\partial \varphi}{\partial \bar{x}}(\tau; \bar{x}) + \mathbf{1}_{\{\tau = \gamma\}} \frac{\partial \varphi}{\partial \tau}(\tau; \bar{x}) \frac{\partial \gamma}{\partial \bar{x}} \right] \nn \\
&\qquad \times (\bar{f}(\bar{x}) + \bar{g}(\bar{x})\hat{\bar{u}}) \geq c.
\end{align}
\end{proposition}

\begin{remark}\rm
Proposition~\ref{prop:dcbf_iff_augmented} implies that the control variable $\hat{\bar{u}}$ satisfies the directional derivative constraint
\begin{equation*}
D_{\bar{f}(\bar{x}) + \bar{g}(\bar{x})\hat{\bar{u}}} \bar\psi_\rmm(\bar{x}) + \alpha_\rmm(\bar\psi_\rmm(\bar{x})) \ge 0
\end{equation*}
if and only if $\hat{\bar{u}}$ satisfies the family of affine constraints \eqref{eq:affine_constraints_new} with $c = -\alpha_\rmm(\bar\psi_\rmm(\bar{x}))$.
\end{remark}

The directional derivative of $\bar \psi_\rmm$ along the trajectories of \cref{eq:augmented,eq:augmented.2,eq:augmented.3} depends on the sensitivity 
\begin{equation*}
\frac{\partial \varphi}{\partial \bar{x}} = \left[\frac{\partial \varphi}{\partial x} \;\frac{\partial \varphi}{\partial \theta} \;\frac{\partial \varphi}{\partial \gamma}\right],
\end{equation*}
where differentiating \eqref{eq:flow_def_shifted} with respect to $x$, $\theta$, and $\gamma$ yields
\begin{align}
\frac{\partial \varphi}{\partial x}(\tau; \bar{x}) &= I + \int_{\gamma}^{\tau} \frac{\partial F}{\partial x}(\varphi(\sigma; \bar{x}), u_\rmp(\sigma; \theta))
 \frac{\partial \varphi}{\partial x}(\sigma; \bar{x}) \, \rmd\sigma,\label{eq:dphi_dx}\\
\frac{\partial \varphi}{\partial \theta}(\tau; \bar{x}) &= \int_{\gamma}^{\tau} \Big[\frac{\partial F}{\partial x}(\varphi(\sigma; \bar{x}), u_\rmp(\sigma; \theta)) \frac{\partial \varphi}{\partial \theta}(\sigma; \bar{x}) \nn \\
&\quad  + \frac{\partial F}{\partial u}(\varphi(\sigma; \bar{x}), u_\rmp(\sigma; \theta)) \frac{\partial u_\rmp}{\partial \theta}(\sigma; \theta)\Big] \, \rmd\sigma,\label{eq:dphi_dtheta}\\
\frac{\partial \varphi}{\partial \gamma}(\tau; \bar{x}) &=\int_{\gamma}^{\tau} \frac{\partial F}{\partial x}(\varphi(\sigma; \bar{x}), u_\rmp(\sigma; \theta)) \frac{\partial \varphi}{\partial \gamma}(\sigma; \bar{x}) \, \rmd\sigma\nn\\
&\quad -F(x, u_\rmp(\gamma;\theta)).\label{eq:dphi_dgamma}
\end{align}
Section~\ref{sec:impl_remark} addresses numerically efficient computation of these sensitivities using the adjoint method \cite{chen2018neural}.

Next, we define a control constraint set that is similar to \cref{eq:def:Kmt} but addresses the planning-time shift.
Let $\alpha_\rmm, \alpha_\theta, \alpha_\gamma$ be extended class-$\SK$ functions, and for all $\bar{x} \in \{\bar{x} \in \bar\Psi \colon \gamma \neq T\}$, consider $\bar K \colon \bar \Psi \rightrightarrows \SU \times \Omega \times \SZ$ defined by
\begin{align}\label{eq:K_constraint}
\bar K(\bar{x}) \triangleq \{&\hat{\bar{u}} \in \SU \times \Omega \times \SZ \colon 
k^\prime(\theta)\,\hat{\omega} + \alpha_\theta(k(\theta)) \geq 0, \nn\\
&D_{\bar{f}(\bar{x}) + \bar{g}(\bar{x})\hat{\bar{u}}} \bar\psi_\rmm(\bar{x}) + \alpha_\rmm(\bar\psi_\rmm(\bar{x})) \geq 0, \nn\\
&L_{\bar{f}} \bar\psi_\rmt(\bar{x}) + L_{\bar{g}} \bar\psi_\rmt(\bar{x}) \hat{\bar{u}} + \alpha_\rmb(\bar\psi_\rmt(\bar{x})) \geq 0, \nn\\
&\hat{z} + \alpha_\gamma(\gamma) \geq 0\}.
\end{align}

Corollary~\ref{cor:pcbf_pair_solution} requires that $(\psi_\rmm,\psi_\rmt)$ is a P-CBF pair to guarantee that the control constraint set $K_{\rmm \rmt}(x,\theta)$ is nonempty. 
Moreover, implementation of the constraint set requires knowledge of specific class-$\SK$ functions that make $K_{\rmm \rmt}(x,\theta)$ nonempty. 
The next result shows that $\bar K(\bar{x})$ is nonempty. 
The result does not require that $(\bar\psi_\rmm, \bar\psi_\rmt)$ is a P-CBF pair, and it holds for any choice of extended class-$\SK$ functions $\alpha_\rmm, \alpha_\theta, \alpha_\gamma$. 

\begin{theorem}\label{thm:feasibility}
    \rm{
    For all $\bar{x} \in \{\bar{x} \in \bar\Psi \colon \gamma \neq T\}$, $\bar K(\bar{x})$ is nonempty. 
    }
\end{theorem}

\begin{proof}
\indent Let $\bar{x}_\rme \in \{\bar{x} \in \bar\Psi \colon \gamma \neq T\}$, and we write its components as $\bar{x}_\rme = [x_\rme^\top \;\theta_\rme^\top \;\gamma_\rme]^\top$. 
It follows from \eqref{eq:def:ubar_p_ubar_b} that $\bar u_\rmp(\bar x_\rme) = [ u_\rme^\top \, \omega_\rme^\top \, z_\rme ]^\top$, where 
\begin{equation*}
u_\rme \triangleq u_\rmp(\gamma_\rme ; \theta_\rme),\quad \omega_\rme = 0_{d\times 1}, \quad z_\rme = 1.
\end{equation*}
We show that $\bar u_\rmp(\bar x_\rme) \in \bar K(\bar x_\rme)$.
To do so, define
\begin{align}
c_1 &\triangleq k^\prime(\theta_\rme) {\omega_\rme} + \alpha_\theta(k(\theta_\rme)), \label{eq:c1}\\
c_2 &\triangleq {z_\rme} + \alpha_\gamma(\gamma_\rme), \label{eq:c2}\\
c_3 &\triangleq L_{\bar{f}} \bar\psi_\rmt(\bar{x}_\rme) + L_{\bar{g}} \bar\psi_\rmt(\bar{x}_\rme) \bar u_\rmp(\bar x_\rme) + \alpha_\rmb(\bar\psi_\rmt(\bar{x}_\rme)), \label{eq:c3}\\
c_4 &\triangleq D_{\bar{f}(\bar{x}_\rme) + \bar{g}(\bar{x}_\rme)\bar u_\rmp(\bar x_\rme)} \bar\psi_\rmm(\bar{x}_\rme) + \alpha_\rmm(\bar\psi_\rmm(\bar{x}_\rme)). \label{eq:c4}
\end{align}
First, since $\bar{x}_\rme \in \bar \Psi$, $\omega_\rme = 0_{d\times 1}$, and $z_\rme = 1$, it follows that 
\begin{equation*}
c_1 =\alpha_\theta(k(\theta_\rme)) \ge 0, \qquad c_2 =1+\alpha_\gamma(\gamma_\rme) \ge 0.
\end{equation*}

To show that $c_3 \ge 0$, define
\begin{equation*}
\eta(\tau) \triangleq \frac{\partial \varphi}{\partial \bar{x}}(\tau;\bar{x}_\rme) [ \bar{f}(\bar{x}_\rme) + \bar{g}(\bar{x}_\rme) \bar u_\rmp(\bar x_\rme) ]
\end{equation*}
and using \cref{eq:dphi_dx,eq:dphi_dtheta,eq:dphi_dgamma,eq:augmented.2,eq:augmented.3} yields
\begin{equation*}
\eta(\tau) = \frac{\partial \varphi}{\partial x}(\tau;\bar{x}_\rme)
F(x_\rme,u_\rmp(\gamma_\rme;\theta_\rme))
+ \frac{\partial \varphi}{\partial \gamma}(\tau;\bar{x}_\rme).
\end{equation*}
Next, \cref{eq:dphi_dx,eq:dphi_dgamma} imply that $\eta(\gamma_\rme) = 0$ and 
\begin{equation}
\frac{\rmd \eta}{\rmd \tau}(\tau)
= \frac{\partial F}{\partial x}(\varphi(\tau;\bar{x}_\rme),u_\rmp(\tau;\theta_\rme))\,\eta(\tau).
\label{eq:eta_ODE}
\end{equation}
Since $\eta(\gamma_\rme) = 0$, it follows from \eqref{eq:eta_ODE} that for all $\tau \in [\gamma_\rme,T]$, $\eta(\tau) = 0$.
Using $\eta(T) = 0$, and $\bar{x}_\rme \in \bar\Psi$, it follows that
\begin{equation*}
c_3 = h_\rmb'(\varphi(T;\bar{x}_\rme)) \eta(T) + \alpha_\rmb(\bar\psi_\rmt(\bar{x}_\rme)) = \alpha_\rmb(\bar\psi_\rmt(\bar{x}_\rme)) \geq 0.
\end{equation*}

To show that $c_4 \ge 0$, it follows from \cref{eq:augmented.2,eq:augmented.3} that $\frac{\partial \gamma}{\partial \bar{x}} (\bar{f}(\bar{x}_\rme) + \bar{g}(\bar{x}_\rme) \bar u_\rmp(\bar x_\rme)) = z_\rme = 1$.
Since, in addition, $\eta(\tau)=0$ for all $\tau \in [\gamma_\rme, T]$, it follows from Proposition~\ref{prop:psi_lipschitz} that 
\begin{align*}
c_4 &= \min_{\tau \in \overline{\ST}(\bar{x}_\rme)} \mathbf{1}_{\{\tau = \gamma_\rme\}}
h_\rms'(\varphi(\tau;\bar{x}_\rme))
F(\varphi(\tau;\bar{x}_\rme),u_\rmp(\tau;\theta_\rme))\\
&\qquad + \alpha_\rmm(\bar\psi_\rmm(\bar{x}_\rme)).
\end{align*}
We consider 2 cases: $\gamma_\rme \in \overline{\ST}(\bar{x}_\rme)$ and $\gamma_\rme \not \in \overline{\ST}(\bar{x}_\rme)$.
First, consider $\gamma_\rme \in \overline{\ST}(\bar{x}_\rme)$, which implies that $\gamma_\rme$ is a minimizer of $h_\rms(\varphi(\tau;\bar{x}_\rme))$ at the left endpoint of $[\gamma_\rme, T]$.
Thus, $h_\rms'(\varphi(\gamma_\rme;\bar{x}_\rme)) F(\varphi(\gamma_\rme;\bar{x}_\rme),u_\rmp(\gamma_\rme;\theta_\rme)) \geq 0$.
Since, in addition, $\bar{x}_\rme \in \bar\Psi$, it follows from \eqref{eq:def:barPsi} that $c_4 \ge \alpha_\rmm(\bar\psi_\rmm(\bar{x}_\rme)) \ge 0$.
Next, consider $\gamma_\rme \not \in \overline{\ST}(\bar{x}_\rme)$, and it follows from \eqref{eq:def:barPsi} that $c_4 = \alpha_\rmm(\bar\psi_\rmm(\bar{x}_\rme)) \ge 0$.

Finally, since $c_1,c_2,c_3,c_4 \ge 0$, it follows from \cref{eq:c1,eq:c2,eq:c3,eq:c4,eq:K_constraint} that $\bar u_\rmp(\bar x_\rme) \in \bar K(\bar x_\rme)$.
\end{proof}

Theorem~\ref{thm:feasibility} shows that for all $\bar{x} \in \{\bar{x} \in \bar\Psi \colon \gamma \neq T\}$, the constraints in \eqref{eq:K_constraint} are feasible. 
The next result can be viewed as an extension of Theorem~\ref{thm:pfcbf_forward_invariance} that uses the planning-time shift to remove the assumption that $(\bar\psi_\rmm, \bar\psi_\rmt)$ is a P-CBF pair.
The result shows that $\bar\Psi$ is made forward invariant by any control selected pointwise from $\bar K(\bar{x})$ for all $\bar{x} \in \{\bar{x} \in \bar\Psi \colon \gamma \neq T\}$ and equal to $\bar{u}_\rmb(\bar{x})$ for all $\bar{x} \in \{\bar{x} \in \bar\Psi \colon \gamma = T\}$.

\begin{theorem}\label{thm:forward_inv_K}
    \rm{
Assume $\SU$ is convex and let Assumption~\ref{ass:backup_set} be satisfied.
Let $\bar{u}_{\rm fi} \colon \bar\Psi \to \SU \times \Omega \times \SZ$ be such that for all $\bar{x} \in \{\bar{x} \in \bar\Psi \colon \gamma \neq T\}$, $\bar{u}_{\rm fi}(\bar{x}) \in \bar K(\bar{x})$, and for all $\bar{x} \in \{\bar{x} \in \bar\Psi \colon \gamma = T\}$, $\bar{u}_{\rm fi}(\bar{x}) = \bar{u}_\rmb(\bar{x})$.
Then, $\bar\Psi$ is forward invariant with respect to \eqref{eq:augmented} with $\bar{u} = \bar{u}_{\rm fi}$.
    }
\end{theorem}

\begin{proof}
\indent Let $\bar{x}_0 \in \bar\Psi$, and consider two cases: (i) $\gamma(t) < T$ for all $t \geq 0$; and (ii) there exists $t_1 \geq 0$ such that $\gamma(t_1) = T$.

For case (i), since $\gamma(t) < T$ for all $t \geq 0$, it follows that $\bar{u}_{\rm fi}(\bar{x}(t)) \in \bar K(\bar{x}(t))$ for all $t \geq 0$.
Since $\bar{u}_{\rm fi}(\bar{x}) \in \bar K(\bar{x})$ satisfies \eqref{eq:K_constraint}, $\bar\psi_\rmm$ is locally Lipschitz and directionally differentiable by Proposition~\ref{prop:psi_lipschitz}, and $\bar\psi_\rmt$ and $k$ are continuously differentiable, it follows from Lemma~\ref{lem:dini_directional} that
\begin{align}
\frac{\rmd^+}{\rmd t}\bar\psi_\rmm(\bar{x}(t))
&\geq -\alpha_\rmm(\bar\psi_\rmm(\bar{x}(t))), \label{eq:proof_psi_c_dini} \\
\frac{\rmd}{\rmd t}\bar\psi_\rmt(\bar{x}(t))
&\geq -\alpha_\rmb(\bar\psi_\rmt(\bar{x}(t))), \label{eq:proof_psi_b_dot} \\
\frac{\rmd}{\rmd t} k(\theta(t))
&\geq -\alpha_\theta(k(\theta(t))), \label{eq:proof_k_dot} \\
z(t) &\geq -\alpha_\gamma(\gamma(t)). \label{eq:proof_gamma_dot}
\end{align}
Let $\eta_\rmm, \eta_\rmt, \eta_\theta, \eta_\gamma \colon [0,\infty) \to \BBR$ satisfy $\dot\eta_\rmm = -\alpha_\rmm(\eta_\rmm)$, $\dot\eta_\rmt = -\alpha_\rmb(\eta_\rmt)$, $\dot\eta_\theta = -\alpha_\theta(\eta_\theta)$, $\dot\eta_\gamma = -\alpha_\gamma(\eta_\gamma)$, where $\eta_\rmm(0) = \bar\psi_\rmm(\bar{x}_0)$, $\eta_\rmt(0) = \bar\psi_\rmt(\bar{x}_0)$, $\eta_\theta(0) = k(\theta_0)$, $\eta_\gamma(0) = \gamma_0$.
Since $\bar{x}_0 \in \bar\Psi$ implies $\eta_\rmm(0) \geq 0$, $\eta_\rmt(0) \geq 0$, $\eta_\theta(0) \geq 0$, and $\eta_\gamma(0) \geq 0$, and each $\alpha_\rmm$, $\alpha_\rmb$, $\alpha_\theta$, $\alpha_\gamma$ is an extended class-$\SK$ function, it follows that for all $t \geq 0$, $\eta_\rmm(t) \geq 0$, $\eta_\rmt(t) \geq 0$, $\eta_\theta(t) \geq 0$, and $\eta_\gamma(t) \geq 0$.
It follows from \cref{eq:proof_psi_c_dini,eq:proof_psi_b_dot,eq:proof_k_dot,eq:proof_gamma_dot} and the comparison lemma \cite[Lemma~3.4]{khalil2002nonlinear} that for all $t \geq 0$, $\bar\psi_\rmm(\bar{x}(t)) \geq \eta_\rmm(t) \geq 0$, $\bar\psi_\rmt(\bar{x}(t)) \geq \eta_\rmt(t) \geq 0$, $k(\theta(t)) \geq \eta_\theta(t) \geq 0$, and $\gamma(t) \geq \eta_\gamma(t) \geq 0$.
Since $\gamma(t) < T$ for all $t \geq 0$, it follows that $\gamma(t) \in [0,T]$ for all $t \geq 0$.
Since $\bar K(\bar{x}) \subseteq \SU \times \Omega \times \SZ$ implies $u(t) \in \SU$ for all $t \geq 0$, it follows that $\bar{x}(t) \in \bar\Psi$ for all $t \geq 0$.

For case (ii), since $\gamma(t) < T$ for all $t \in [0, t_1)$, it follows from case (i) that for all $t \in [0, t_1)$, $\bar{x}(t) \in \bar\Psi$.
Since $\bar{x}(t_1) \in \bar\Psi$ and $\gamma(t_1) = T$, it follows from \eqref{eq:psi_b} that $h_\rmb(x(t_1)) = \bar\psi_\rmt(\bar{x}(t_1)) \ge 0$, which implies $x(t_1) \in \SC_\rmb$.
Since $x(t_1) \in \SC_\rmb$ and $\bar{u}_{\rm fi}(\bar{x}) = \bar{u}_\rmb(\bar{x})$ for $\gamma = T$, it follows from \eqref{eq:def:ubar_p_ubar_b} that $u(t) = u_\rmb(x(t))$ for all $t \ge t_1$.
Since $x(t_1) \in \SC_\rmb$, Proposition~\ref{prop:backup_controller_exists} implies that for all $t \in [t_1,\infty)$, $x(t) \in \SC_\rmb$.
Since, in addition, $\SC_\rmb \subset \SC_\rms$, and for all $t \in [t_1,\infty)$, $\gamma(t) = T$, it follows from \cref{eq:psi_min,eq:psi_b} that for all $t \in [t_1,\infty)$, $\bar\psi_\rmm(\bar{x}(t)) = h_\rms(x(t)) \ge 0$ and $\bar\psi_\rmt(\bar{x}(t)) = h_\rmb(x(t)) \ge 0$.
Thus, for all $t \geq 0$, $\bar{x}(t) \in \bar\Psi$.
\end{proof}

\subsection{OPTIMIZATION-BASED CONTROL}

This section presents the safe optimal flow control that optimizes the integral cost while guaranteeing that $\bar{x}(t) \in \bar\Psi$ for all $t \geq 0$.
Consider the cost $\bar J \colon \BBR^{\bar n} \to \BBR$ given by
\begin{equation}\label{eq:J_augmented}
\bar J(\bar{x}) \triangleq W \left ( \varphi(T; \bar{x}) \right )
 + \int_{\gamma}^{T} R\left(\varphi(\tau; \bar{x}),u_\rmp(\tau;\theta) \right) \, {\rm d} \tau,
\end{equation}
which is analogous to \eqref{eq:def:J} except $\phi$ is replaced by $\varphi$. 
For $\gamma = 0$, $\bar J(\bar{x})$ reduces to $J(x,\theta)$.
To make the time derivative of $\bar J$ small, we consider a quadratic cost that is analogous to \eqref{eq:SJ}. 
Specifically, consider the quadratic cost
\begin{align}\label{eq:qp_cost}
\bar\SJ(\hat{\bar{u}}; \bar{x}) &\triangleq \frac{\partial \bar J}{\partial x} g(x) \hat{u} + \frac{\partial \bar J}{\partial \theta} \hat{\omega} + \frac{\partial \bar J}{\partial \gamma} \hat{z} \nn \\
&\quad + (\hat{u} - u_\rmp(\gamma; \theta))^\top Q_u (\hat{u} - u_\rmp(\gamma; \theta)) \nn \\
&\quad + \hat{\omega}^\top Q_\omega \hat{\omega} + Q_z \hat{z}^2 + \lambda \hat{z},
\end{align}
where $Q_z > 0$ and $\lambda \geq 0$.
Similar to \eqref{eq:SJ}, the first 3 terms make $\frac{\rmd \bar J}{\rmd t}$ small, and the next 3 terms provide regularization that make \eqref{eq:qp_cost} strictly convex.
The final term $\lambda \hat{z}$ penalizes increasing $\gamma$. 
In other words, this incentivizes a large planning horizon $T-\gamma$. 

\begin{remark}\rm
Since $\bar J$ involves an integral over $[\gamma, T]$, it follows that increasing $\gamma$ shrinks the prediction window, which can decrease $\bar J$.
Consequently, the term $\frac{\partial \bar J}{\partial \gamma} \hat{z}$ in \eqref{eq:qp_cost} can incentivize increasing $\gamma$ to reduce $\bar J$. 
This effect is not generally desirable; rather it is desirable for $\gamma$ to increase if and only if needed for feasibility. 
This work includes the $\lambda \hat{z}$ with relatively large $\lambda$ to mitigate this effect. 
The effect can also be mitigated by omitting $\frac{\partial \bar J}{\partial \gamma} \hat{z}$ from \eqref{eq:qp_cost} and/or multiplying the integral in \eqref{eq:J_augmented} by $\frac{T}{T-\gamma}$ to provide normalization. 
All 3 methods are effective in simulation. 
\end{remark}

The next result extends Proposition~\ref{prop:K_convex} and shows that $\bar\SJ$ has a unique global minimizer over $\bar K(\bar{x})$. 
The proof is similar to that of Proposition~\ref{prop:K_convex}.

\begin{proposition}\label{prop:K_bar_mt_convex}
\rm
Assume $\SU$, $\Omega$, and $\SZ$ are convex.
Then, for all $\bar{x} \in \{\bar{x} \in \bar\Psi \colon \gamma \neq T\}$, the following hold:
\begin{enumerate}[leftmargin=0.7cm]

\item\label{prop:K_bar_mt_convex.b} $\operatorname*{argmin}_{\hat{\bar{u}} \in \bar K(\bar{x})} \bar\SJ(\hat{\bar{u}}; \bar{x})$ exists and is unique.

\item\label{prop:K_bar_mt_convex.a} $\bar K(\bar{x})$ is convex.

\end{enumerate}
\end{proposition}

For all $\bar{x} \in \{\bar{x} \in \bar\Psi \colon \gamma \neq T\}$, define
\begin{equation}\label{eq:v_star_def}
\bar{u}_{\rmp*}(\bar x) \triangleq \operatorname*{argmin}_{\hat{\bar{u}} \in \bar K(\bar{x})} \bar\SJ(\hat{\bar{u}}; \bar{x}),
\end{equation}
where it follows from Proposition~\ref{prop:K_bar_mt_convex} that $\bar{u}_{\rmp*}(\bar{x})$ exists and is unique.

\begin{remark}\rm
Similar to Remark~\ref{remark:PCBF-QP}, if $\SU$, $\Omega$, and $\SZ$ are convex polyhedra, then Proposition~\ref{prop:dcbf_iff_augmented} implies that all constraints that define $\bar K(\bar{x})$ are affine.
In this case, the optimal control \eqref{eq:v_star_def} can be obtained efficiently from a QP that depends on the predicted flow $\varphi$.
Section~\ref{sec:impl_remark} presents implementation of this QP.
\end{remark}

Finally, consider the safe optimal flow control $\bar{u}_* \colon \bar\Psi \to \SU \times \Omega \times \SZ$ defined by 
\begin{equation}\label{eq:control_law}
\bar{u}_*(\bar{x}) \triangleq \begin{cases}
\bar{u}_{\rmp*}(\bar{x}), & \gamma \in[0, T), \\
\bar{u}_\rmb(\bar x), & \gamma = T.
\end{cases}
\end{equation}
The term $\lambda \hat z$ in \eqref{eq:qp_cost} incentivizes $\gamma$ to be close to zero; however, it does not prevent $\gamma=T$.
The constraint set $\bar K(\bar x)$ is not well defined for $\gamma=T$ because the directional derivative of $\bar \psi_\rmm$ does not necessarily exist.
Hence, if $\gamma = T$, then control $\bar{u}_*$ switches to $\bar{u}_\rmb$.

The next result shows that the safe optimal flow control $\bar{u}_*$ makes $\bar \Psi$ forward invariant. 
This result is an immediate consequence of Theorem~\ref{thm:forward_inv_K} because $\bar{u}_*(\bar{x}) \in \bar K(\bar{x})$ for all $\bar{x} \in \{\bar{x} \in \bar\Psi \colon \gamma \neq T\}$ and $\bar{u}_*(\bar{x}) = \bar{u}_\rmb(\bar{x})$ for all $\bar{x} \in \{\bar{x} \in \bar\Psi \colon \gamma = T\}$.

\begin{corollary}\label{cor:main_safety}
    \rm{
Assume $\SU$, $\Omega$, and $\SZ$ are convex, and assume Assumption~\ref{ass:backup_set} is satisfied.
Then, $\bar\Psi$ is forward invariant with respect to \eqref{eq:augmented} with $\bar{u} = \bar{u}_*$.
    }
\end{corollary}

Corollary~\ref{cor:main_safety} implies that for all $(t,\tau) \in [0,\infty) \times [\gamma(t), T]$, $\varphi(\tau;\bar{x}(t)) \in \SC_\rms$; and for all $t \geq 0$, $\varphi(T;\bar{x}(t)) \in \SC_\rmb$ and $\theta(t) \in \Theta$.
Thus, $\bar{u}_*$ guarantees that the predicted flow $\varphi(\cdot;\bar{x}(t))$ and actual state $x(t)$ are in the safe set $\SC_\rms$ for all time $t \ge 0$, while minimizing $\bar\SJ$, which aims to decrease the receding-horizon cost $\bar J$ along the trajectories of \eqref{eq:augmented}.

\section{QP IMPLEMENTATION OF SAFE OPTIMAL FLOW CONTROL} \label{sec:impl_remark}

This section presents a QP implementation of the safe optimal flow control \eqref{eq:control_law}.
For this section, we assume $\SU$, $\Omega$, and $\SZ$ are convex polyhedra.

\subsection{NUMERICAL IMPLEMENTATION}\label{sec:impl_numerical}

Proposition~\ref{prop:dcbf_iff_augmented} implies that the constraint on $\bar\psi_\rmm$ in \eqref{eq:K_constraint} is equivalent to the family of affine constraints \eqref{eq:affine_constraints_new}.
Thus, \eqref{eq:v_star_def} can be expressed as 
\begin{subequations}\label{eq:semi_infinite_qp}
\begin{align}
&\quad \bar{u}_{\rmp*}(\bar{x}) = \operatorname*{argmin}_{\hat{\bar{u}} \in \SU \times \Omega \times \SZ} \bar\SJ(\hat{\bar{u}}; \bar{x}) \label{eq:semi_infinite_qp_obj}\\
&\text{subject to} \nn \\
&\quad k^\prime(\theta)\,\hat{\omega} + \alpha_\theta(k(\theta)) \geq 0, \label{eq:semi_infinite_qp_param}\\
&\quad h_\rms'(\varphi(\tau; \bar{x})) \Big[ \frac{\partial \varphi}{\partial \bar{x}}(\tau; \bar{x}) + \mathbf{1}_{\{\tau = \gamma\}} \frac{\partial \varphi}{\partial \tau}(\tau; \bar{x}) \frac{\partial \gamma}{\partial \bar{x}} \Big] \nn\\
&\qquad \times (\bar{f}(\bar{x}) + \bar{g}(\bar{x})\hat{\bar{u}}) + \alpha_\rmm(\bar\psi_\rmm(\bar{x})) \geq 0, \quad \forall \tau \in \overline{\ST}(\bar{x}), \label{eq:semi_infinite_qp_safe}\\
&\quad L_{\bar{f}} \bar\psi_\rmt(\bar{x}) + L_{\bar{g}} \bar\psi_\rmt(\bar{x}) \hat{\bar{u}} + \alpha_\rmb(\bar\psi_\rmt(\bar{x})) \geq 0, \label{eq:semi_infinite_qp_term}\\
&\quad \hat{z} + \alpha_\gamma(\gamma) \geq 0. \label{eq:semi_infinite_qp_gamma}
\end{align}
\end{subequations}
All constraints in \eqref{eq:semi_infinite_qp} are affine.
However, \eqref{eq:semi_infinite_qp_safe} may constitute infinitely many affine constraints because $\overline{\ST}(\bar{x})$ may contain infinitely many points.
Thus, \eqref{eq:semi_infinite_qp} is a semi-infinite QP, which can be solved with a variety of approaches; see \cite{lopez2007semi}.
This article uses a discretization method over the planning time, where $\overline{\ST}(\bar{x})$ is approximated with a finite set. 

Let $N$ be a positive integer.
Define $T_\rmd \triangleq \frac{T - \gamma}{N}$ and $\tau_i \triangleq \gamma + i T_\rmd$ for all $i \in \{0, 1, \ldots, N\}$.
Then, consider the set
\begin{equation}\label{eq:T_e}
\overline{\ST}_\rme(\bar{x}) \triangleq \operatorname*{argmin}_{\hat\tau \in \{\tau_0, \ldots, \tau_N\}} h_\rms(\varphi(\hat\tau; \bar{x})),
\end{equation}
which contains the discrete prediction times at which $h_\rms(\varphi(\,\cdot\,;\bar{x}))$ attains its minimum.
The approximation $\overline{\ST}_\rme(\bar{x})$ converges to $\overline{\ST}(\bar{x})$ as $N \to \infty$.
Furthermore, replacing $\overline{\ST}(\bar{x})$ with $\overline{\ST}_\rme(\bar{x})$ in \eqref{eq:semi_infinite_qp_safe} yields a QP with a finite number of affine constraints.

\begin{remark}\rm
The approximation \eqref{eq:T_e} can be improved by using the sign changes of $\frac{\rmd}{\rmd \tau} h_\rms(\varphi(\tau; \bar{x})) = h_\rms^\prime(\varphi(\tau; \bar{x})) F(\varphi(\tau; \bar{x}), u_\rmp(\tau; \theta))$ evaluated at $\tau \in \{\tau_0, \ldots, \tau_N\}$ to bracket local critical points.
Then, the discretization time step can be refined (e.g., midpoint refinement or bisection) in the appropriate intervals.
\end{remark}
 
\begin{remark}\rm
In practice, $h_\rms(\varphi(\,\cdot\,;\bar{x}))$ often has a unique global minimizer over $[\gamma, T]$. 
In this case, \eqref{eq:semi_infinite_qp_safe} reduces to a single affine constraint.
\end{remark}

The QP \eqref{eq:semi_infinite_qp} requires the gradients 
\begin{equation}\label{eq:gradients}
\frac{\partial \bar J(\bar x)}{\partial \bar x}, \qquad \frac{\partial \bar\psi_\rmt(\bar x)}{\partial \bar x}, \quad \frac{\partial h_\rms(\varphi(\tau;\bar{x}))}{\partial \bar x},    
\end{equation}
for $\tau \in \overline{\ST}_\rme(\bar{x})$. 
These can be computed using a forward sensitivity approach, where $\frac{\partial \varphi}{\partial \bar{x}}(\tau; \bar{x})$ is obtained by solving \eqref{eq:dphi_dx}--\eqref{eq:dphi_dgamma} forward in prediction time alongside the predicted flow \eqref{eq:flow_def_shifted}. 
However, this approach requires integrating a system of ordinary differential equations with dimension $n(n+d+1)$, which is quadratic in $n$ and scales with the number of plan parameters $d$.

Alternatively, the adjoint approach \cite{chen2018neural} is a more computationally efficient method to compute \eqref{eq:gradients}.
The advantage is that each adjoint ordinary differential equation has dimension $n$.
Specifically, since $\bar J$, $\bar \psi_\rmt$, and $h_\rms(\varphi(\tau;\bar x))$ are scalar functions of $\bar x$, each gradient in \eqref{eq:gradients} can be computed with an adjoint $n$-dimensional backward integration. The gradients of $\bar J$ and $\bar\psi_\rmt$ require backward integration from $T$ to $\gamma$, while the gradient of $h_\rms(\varphi(\tau;\bar{x}))$ at each $\tau \in \overline{\ST}_\rme(\bar{x})$ requires backward integration from $\tau$ to $\gamma$.
Thus, all gradients are computed by integrating $2 + n_\rme$ different $n$-dimensional ordinary differential equations, where $n_\rme$ is the number of elements in $\overline{\ST}_\rme(\bar{x})$.
Hence, the complexity is linear in $n$ and does not increase with $d$, which implies that the adjoint method has significant computational benefit for large $n$ and/or large $d$. 
We also note that the $2 + n_\rme$ differential equations are decoupled from one another and can be solved in parallel.

Algorithm~\ref{alg:flowbarrier} summarizes \textit{FlowBarrier}, which is the QP implementation of the safe optimal flow control  \eqref{eq:control_law} where $\delta t >0$ is the time increment for a zero-order hold on the control $\bar u_*$. 
We write the components of $\bar u_*$ as $\bar u_* = \matls u_* \\ \omega_* \\ z_* \matrs$. 
At each time step, Algorithm~\ref{alg:flowbarrier} has 4 steps: (i) forward integration of \eqref{eq:flow_def_shifted} to obtain the predicted flow $\varphi$, and evaluate $\bar\psi_\rmm$, $\bar\psi_\rmt$, and $k$; (ii) backward integration of the $2+n_\rme$ parallel adjoint ordinary differential equations to obtain \eqref{eq:gradients}; (iii) solve QP \eqref{eq:semi_infinite_qp} with $\overline{\ST}_\rme$ replacing $\overline{\ST}$ to obtain $\bar{u}_{\rmp*}$; and (iv) execute control $u_*$ and update $\theta$ and $\gamma$ by integrating optimal derivatives $\omega_*$ and $z_*$ over time increment $\delta t$.
If $\gamma = T$, then steps (ii) and (iii) are replaced by solving the QP \eqref{eq:backup_controller} to obtain the backup $\bar u_\rmb$.

\begin{algorithm}[ht!]
\caption{FlowBarrier}
\DontPrintSemicolon
\small\linespread{0.9}\selectfont
\label{alg:flowbarrier}
\SetKwComment{Comment}{$\triangleright$\ }{}
\SetKwInput{KwParam}{Parameters}
\KwParam{\small$h_\rmb, h_\rms, k, \beta_i, p, N, Q_u, Q_\omega, Q_z, \lambda, \alpha_\rmm, \alpha_\rmb, \alpha_\theta, \alpha_\gamma, \delta t$} 
\KwInit{$\theta \gets \theta_0$, $\gamma \gets \gamma_0$}
\BlankLine
\For{$j = 0, 1, 2, \ldots$}{
    \tcp{Measure state}
    $x \gets x(j\delta t)$\;
    $\bar{x} \gets [x^\top \;\theta^\top \;\gamma]^\top$\;
    \tcp{Forward pass}
    Solve \eqref{eq:flow_def_shifted} for $\{\varphi(\tau_i;\bar x)\}_{i=0}^N$\;
    $\bar\psi_\rmm \gets \eqref{eq:psi_min}$, $\bar\psi_\rmt \gets \eqref{eq:psi_b}$, $\overline{\ST}_\rme \gets \eqref{eq:T_e}$\;
    \eIf{$\gamma = T$}{
        $u_* \gets$ solution to \eqref{eq:backup_controller}\;
        $\omega_* \gets 0_{d \times 1}$\;
        $z_* \gets 0$\;
    }{
        \tcp{Backward pass}
        Compute $\frac{\partial \bar J}{\partial \bar x}$, $\frac{\partial \bar\psi_\rmt}{\partial \bar x}$, and $\frac{\partial}{\partial \bar x} h_\rms(\varphi(\tau;\bar{x}))$ for all $\tau \in \overline{\ST}_\rme$\;
        \tcp{Solve QP}
        $[u_*^\top \;\; \omega_*^\top \;\; z_*]^\top \gets$ solution to \eqref{eq:semi_infinite_qp} with $\overline{\ST}_\rme$ replacing $\overline{\ST}$\;
    }
    \tcp{Execute control}
    $u \gets u_*$\;
    \tcp{Update control plan}
    $\theta \gets \theta + \omega_* \delta t$\;
    $\gamma \gets \gamma + z_* \delta t$\;
}
\end{algorithm}

\begin{remark}\label{rem:recovery}\rm
Although external disturbances, model uncertainties, or sampled-data effects can cause $\bar{x}$ to leave $\bar\Psi$, it follows from Theorem~\ref{thm:feasibility} that $\{\bar{x} \in \bar\Psi \colon \gamma \neq T\} \subseteq \{\bar{x} \colon \bar K(\bar{x}) \neq \emptyset\}$.
Thus, \eqref{eq:semi_infinite_qp} may be feasible for $\bar{x} \notin \bar\Psi$.
If $\bar{x} \notin \bar\Psi$ and $\bar K(\bar{x})$ is nonempty, then $\bar{u}_{\rmp*}(\bar{x})$ drives $\bar{x}$ back to $\bar\Psi$.
If $\bar K(\bar{x})$ is empty, then slack variables can be used to ensure the QP is feasible while trying to drive $\bar{x}$ back to $\bar\Psi$.
\end{remark}

\begin{remark}\label{rem:gamma_T}\rm
If there is a time $t_1 \ge 0$ such that $\gamma(t_1) = T$, then the control \eqref{eq:control_law} results in $u_* = u_\rmb$ and $\gamma = T$ for all $t \ge t_1$.
Thus, \eqref{eq:control_law} does not have a mechanism to decrease $\gamma$ and recover the prediction window. 
Simulations suggest that it may be unlikely that $\gamma$ increases to $T$; however, if this occurs, then there is a practical approach to recover the prediction window. 
To explain, consider the QP
\begin{subequations}\label{eq:rescue_qp}
\begin{align}
&\quad \min_{\substack{\hat{\bar{u}} \in \SU \times \Omega \times \SZ \\ \hat\delta_\rmm,\, \hat\delta_\rmt \in \BBR}} \bar\SJ(\hat{\bar{u}}; \bar{x}) + Q_\delta (\hat\delta_\rmm^2 + \hat\delta_\rmt^2) + \lambda_\delta (\hat\delta_\rmm + \hat\delta_\rmt) \label{eq:rescue} \\
&\text{subject to} \nn \\
&\quad L_f h_\rmb(x) + L_g h_\rmb(x) \hat{u} + \alpha_\rmb(h_\rmb(x)) \geq 0, \label{eq:rescue_c1} \\
&\quad k^\prime(\theta)\,\hat{\omega} + \alpha_\theta(k(\theta)) \geq 0, \label{eq:rescue_c2} \\
&\quad D_{\bar{f}(\bar{x}) + \bar{g}(\bar{x})\hat{\bar{u}}} \bar\psi_\rmm(\bar{x}) + \hat\delta_\rmm \geq 0, \label{eq:rescue_c3} \\
&\quad L_{\bar{f}} \bar\psi_\rmt(\bar{x}) + L_{\bar{g}} \bar\psi_\rmt(\bar{x}) \hat{\bar{u}} + \hat\delta_\rmt \geq 0, \label{eq:rescue_c4} \\
&\quad \hat{z} + \alpha_\gamma(\gamma) \geq 0, \label{eq:rescue_c5}
\end{align}
\end{subequations}
where $Q_\delta > 0$ and $\lambda_\delta \ge 0$.
For all $\bar{x} \in \bar\SC_\rmb \triangleq \{\bar{x} \colon x \in \SC_\rmb,\, \theta \in \Theta,\, \gamma \in [0,T)\}$, the constraints \cref{eq:rescue_c1,eq:rescue_c2,eq:rescue_c3,eq:rescue_c4,eq:rescue_c5} are feasible. 
Specifically, $\bar{x} \in \bar\SC_\rmb$, $\hat\omega = 0$ and $\hat{z} = 0$ satisfy \cref{eq:rescue_c2,eq:rescue_c5}; Assumption~\ref{ass:backup_set} implies that there exists $\hat{u} \in \SU$ satisfying \eqref{eq:rescue_c1}; and the slack variables $\hat\delta_\rmm$ and $\hat\delta_\rmt$ can be selected to satisfy \cref{eq:rescue_c3,eq:rescue_c4}. 
Since, in addition, \eqref{eq:rescue} is strictly convex, it follows that for all $\bar x \in \bar \SC_\rmb$, the QP \eqref{eq:rescue_qp} has a unique solution, which we denote by $\left(\bar{u}_\rmr(\bar{x}),\delta_\rmm(\bar{x}),\delta_\rmt(\bar{x}) \right )$.

If there is a time $t_1 \ge 0$ such that $\gamma(t_1) = T$, then the prediction window may be recovered with the following procedure.
First, set $\gamma(t_1) = 0$, select $\theta(t_1) \in \Theta$, and let $\bar{u}(\bar{x}) = \bar{u}_\rmr(\bar{x})$, which attempts to drive $\bar{x}$ to $\bar\Psi$ while making $\SC_\rmb$ forward invariant. 
Next, if there exists a time $t_2 > t_1$ such that $\bar{x}(t_2) \in \bar\Psi$, then switch the control back to \eqref{eq:control_law}. 
\end{remark}

\subsection{SOFT-MINIMUM CONSTRUCTION FOR $h_\rms$, $h_\rmb$, AND $k$} \label{sec:impl_softmin}

In this article, the safe set $\SC_\rms$, backup set $\SC_\rmb$, and admissible parameter set $\Theta$ are each defined as the zero-superlevel set of one function (i.e., $h_\rms$, $h_\rmb$, and $k$, respectively). 
In practice, it can be useful to define each set as the intersection of zero-superlevel sets of multiple functions.
The approaches and analysis in this article extend directly to the case where $\SC_\rms$, $\SC_\rmb$, and $\Theta$ are the intersection of zero-superlevel sets of multiple functions. 
In this case, $\bar K$ includes a constraint for each function, which increases complexity of the QP \eqref{eq:v_star_def}.
An alternative approach is to use the log-sum-exponential soft minimum 
\cite{backupautomatic,rabiee2024closed,compositionACC} to compose multiple barrier functions into a single one. 
Specifically, let $\rho > 0$, and consider $\mbox{softmin}_\rho \colon \BBR^{n_{\rm sm}} \to \BBR$ defined by
\begin{equation}\label{eq:softmin_def}
\mbox{softmin}_\rho(z_1, \ldots, z_{n_{\rm sm}}) \triangleq -\frac{1}{\rho} \ln  \sum_{i=1}^{n_{\rm sm}} e^{-\rho z_i}.
\end{equation}
The soft minimum \eqref{eq:softmin_def} provides a continuously differentiable lower bound on the minimum (e.g., \cite{backupautomatic,rabiee2024closed}), that is, $\mbox{softmin}_\rho(z_1, \ldots, z_{n_{\rm sm}}) \leq \min \{z_1, \ldots, z_{n_{\rm sm}} \}$.

To illustrate a soft minimum of $h_\rms$, consider $n_\rms$ continuously differentiable barrier functions $b_1, \ldots, b_{n_\rms} \colon \BBR^n \to \BBR$.
The set where all constraints are satisfied is 
\begin{equation*}
\SSS_\rms \triangleq \{ x \in \BBR^n \colon b_1(x) \geq 0, \ldots, b_{n_\rms}(x) \geq 0 \},
\end{equation*}
which is the intersection of the zero-superlevel sets of $b_1, \ldots, b_{n_\rms}$.
Then, consider the safe set $\SC_\rms$, where 
\begin{equation}\label{eq:h_s_construction}
h_\rms(x) = \mbox{softmin}_{\rho_\rms}(b_1(x), \ldots, b_{n_\rms}(x)),
\end{equation}
and it follows from \cite{safari2024TSCT} that $\SC_\rms \subseteq \SSS_\rms$, and $\SC_\rms \to \SSS_\rms$ as $\rho_\rms \to \infty$.
Furthermore, the worst case conservativeness of the soft-minimum approximation of the minimum is $(\ln n_\rms)/\rho_\rms$ \cite{safari2024TSCT}.
Thus, $\rho_\rms$ can be selected to limit conservativeness based on $n_\rms$.
If $\rho_\rms$ is small, then the soft minimum is a conservative approximation of the minimum.
However, if $\rho_\rms$ is large, then the magnitude of $h_\rms^\prime$ is large at points where the minimum is not differentiable.
Thus, selecting $\rho_\rms$  is a trade-off between the conservativeness of $\SC_\rms$ and the magnitude of $h^\prime_\rms$. 
A similar approach can be used to construct $h_\rmb$.

The soft minimum can also be used to construct $k$. 
Since $\SU$ is a convex polyhedron, it can be expressed as
\begin{equation*}
\SU = \{u \in \BBR^m \colon a_1^\top u + d_1 \geq 0, \ldots,  a_r^\top u + d_r \geq 0\},
\end{equation*}
where $a_1, \ldots, a_r \in \BBR^m$ and $d_1, \ldots, d_r \in \BBR$.
Let the control plan $u_\rmp$ be given by \eqref{eq:u_p_construction}, which implies that for each $\tau \in [0,T]$, $u_\rmp(\tau;\theta)$ is a convex combination of $\theta_1,\ldots,\theta_p$. 
Hence, if $(\theta_1,\ldots,\theta_p) \in \SU^p$, then for all $\tau \in [0,T]$, $u_\rmp(\tau;\theta) \in \SU$. 
Thus, the admissible parameter set $\Theta$ can be constructed with
\begin{align}\label{eq:k_construction}
k(\theta) = \mbox{softmin}_{\rho_k}(& a_1^\top \theta_1 + d_1, \ldots, a_r^\top \theta_1 + d_r, \ldots, \nn\\
&a_1^\top \theta_p + d_1, \ldots, a_r^\top \theta_p + d_r),
\end{align}
where $\rho_k >0$. 
Similar to above, it follows from \cite{safari2024TSCT} that $\Theta \subseteq \SU^p$, and $\Theta \to \SU^p$ as $\rho_k \to \infty$.

\section{APPLICATION TO A GROUND ROBOT}\label{sec:application}

\begin{figure*}[t!]
\centering
\includegraphics[width=\textwidth]{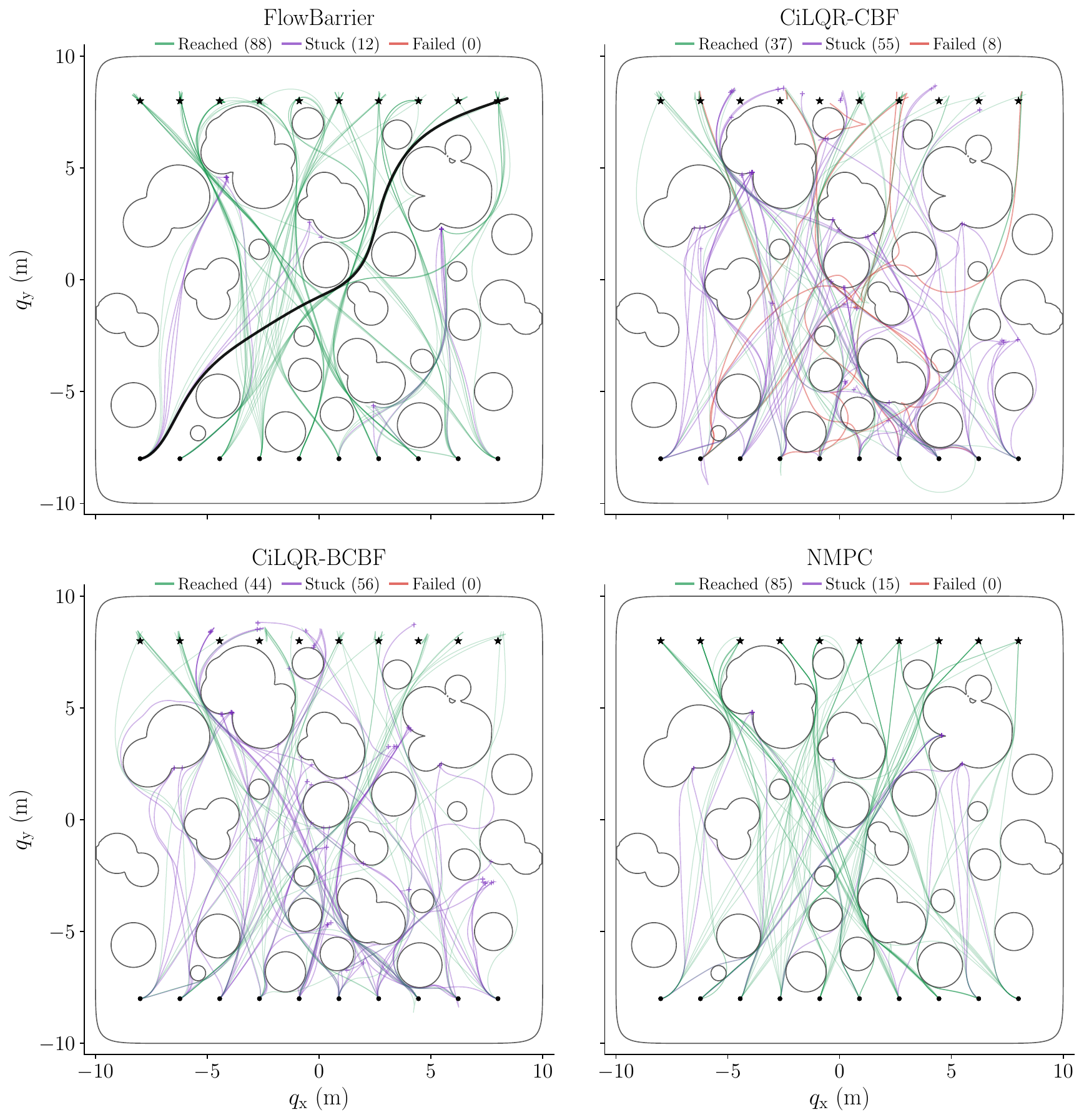}
\caption{Closed-loop trajectories for FlowBarrier, NMPC, CiLQR-CBF, and CiLQR-BCBF across 100 navigation trials in a dense obstacle environment. Trajectories are color-coded by task outcome: reached (green), stuck (purple), and failed (red). The highlighted black trajectory in FlowBarrier corresponds to the navigation task from initial state $x_0 = [\,-8\; -8\; 0 \; 0 \,]^\top$ to goal state $x_\rmd = [\,8\; 8 \; 0 \; 0\,]^\top$.}
\label{fig:monte_carlo}
\end{figure*}

Consider the nonholonomic ground robot modeled by \eqref{eq:affine_control}, where
\begin{equation*}
    f(x) = \begin{bmatrix}
     v\cos\vartheta \\
     v\sin\vartheta \\
    0 \\
    0
    \end{bmatrix},
    \,
    g(x) = \begin{bmatrix}
    0 & 0\\
    0 & 0\\
    1 & 0 \\
    0 & 1
    \end{bmatrix},
    \,
    x = \begin{bmatrix}
    q_\rmx\\
    q_\rmy\\
    v\\
    \vartheta
    \end{bmatrix}, 
    \,
    u = \begin{bmatrix}
    u_1\\
    u_2
    \end{bmatrix}, 
\end{equation*}
and $q \triangleq [ \, q_\rmx \quad q_\rmy \, ]^\top$ is the robot's position in an orthogonal coordinate system, $v$ is the speed, and $\vartheta$ is the direction of the velocity vector (i.e., the angle from $[ \, 1 \quad 0 \, ]^\top$ to $[ \, \dot q_\rmx \quad \dot q_\rmy \, ]^\top$).
Let $\bar u_1 = 2$, $\bar u_2 = 1$, and
\begin{equation*}
\SU = \{ [ u_1 \, u_2 ]^\top \in\BBR^2\colon u_1 \in [-\bar u_1, \bar u_1], u_2 \in [-\bar u_2, \bar u_2] \}.
\end{equation*}
Since $\SU$ is a convex polyhedron, the admissible parameter set $\Theta$ is constructed using \eqref{eq:k_construction} with $\rho_k = 50$.
Let $\Omega = \BBR^d$, which is convex and satisfies $0 \in \Omega$.
Let $\SZ = (-\infty, 1]$, which is a convex polyhedron in $\BBR$ and satisfies $1 \in \SZ$.
The upper bound $z \leq 1$ ensures that $\dot\gamma = z \leq 1$, so the planning-time shift $\gamma$ does not advance faster than real time.

Consider the map shown in \Cref{fig:monte_carlo}, which has $46$ circles and a wall. 
For $i \in \{1, \ldots, 46 \}$, the area outside the $i$th obstacle is modeled as the zero-superlevel set of
\begin{equation*}
b_i(x) = \left\| q - c_i \right\| - r_i,
\end{equation*}
where $c_i \in \BBR^2$ and $r_i > 0$ are the center and radius of the $i$th circle.
Similarly, the area inside the wall is modeled as the zero-superlevel set of
\begin{equation*}
    b_{47}(x) = a - \left(|q_\rmx|^{20} + |q_\rmy|^{20}\right)^{1/20},
\end{equation*}
where $a > 0$ specifies the half-width of the square wall.
The bounds on speed $v$ are modeled as the zero-superlevel sets of
\begin{equation*}
    b_{48}(x) = 2 - v, \quad b_{49}(x) = v + 2.
\end{equation*}
The safe set $\SC_\rms$ is given by \eqref{eq:safe_set} and \eqref{eq:h_s_construction} with $n_\rms = 49$ and $\rho_\rms = 20$.
The safe set $\SC_\rms$ projected into the $q_\rmx$-$q_\rmy$ plane is shown in \Cref{fig:monte_carlo}. Note that $\SC_\rms$ is also bounded in speed $v$, specifically, for all $x \in \SC_\rms$, $v \in [-2,2]$.

The backup safe set $\SC_\rmb$ is given by \eqref{eq:backup_set}, where
\begin{equation}\label{eq:backup_set_robot}
    h_\rmb(x) = h_\rms(x) - \frac{v^2}{2\bar u_1}.
\end{equation}
To verify Assumption~\ref{ass:backup_set}, consider the control $\tilde{u}_\rmb(x) \triangleq [\,-\bar u_1 \operatorname{sgn}(v) \;\; 0\,]^\top$, which applies maximum deceleration.
It can be shown by direct computation that $L_f h_\rmb(x) + L_g h_\rmb(x) \tilde{u}_\rmb(x) \geq 0$ for all $x \in \SC_\rmb$.
Thus, $\tilde{u}_\rmb(x) \in K_\rmb(x)$ for any extended class-$\SK$ function $\alpha_\rmb$, which implies that Assumption~\ref{ass:backup_set} is satisfied.

The control objective is for the robot to move from its initial state to a desired state $x_\rmd \triangleq \begin{bmatrix} q_{\rmd,\rmx}\; q_{\rmd,\rmy} \; 0 \; 0   \end{bmatrix}^\top\in \BBR^4$ without violating safety (i.e., hitting an obstacle or violating speed bounds). To accomplish this objective, consider the cost function given by \eqref{eq:J_augmented}, where
\begin{equation}\label{eq:cost_functions}
R(x,u) = W(x) = \| x - x_\rmd \|^2.
\end{equation}
Minimizing the cost function drives the state toward the desired state $x_\rmd$. 
The quadratic program \eqref{eq:v_star_def} includes the cost gradient as a linear term in the objective, which encourages $\bar{u}_{\rmp*}$ to reduce the cost along the predicted trajectory. 
Thus, minimizing the quadratic program objective drives the state toward the minimizer of $\bar J$ while satisfying the safety constraints. 
This approach eliminates the need for an explicit reference control, unlike standard CBF methods where a desired control input must be specified.

We implement Algorithm~\ref{alg:flowbarrier} with $\alpha_\rmm(\bar\psi_\rmm) = 10\bar\psi_\rmm$, $\alpha_\rmb(\bar\psi_\rmt) = \bar\psi_\rmt$, $\alpha_\theta(k) = 10k$, $\alpha_\gamma(\gamma)= 0.1\gamma$, $T = 4\, \rms$, $T_\rmd = 0.05\, \rms$, $p=80$, $Q_u = 10^6I_m$, $Q_\omega = 30I_{d}$, $Q_z = 10^{-6}$, $\lambda = 1000$, $N = 80$, $\delta t = 0.005\, \rms$, and $u_\rmp$ given by \Cref{ex:foh}, where $\beta_1,\ldots,\beta_p$ are the degree-one B-splines \cref{eq:u_p_construction.b,eq:u_p_construction.a}.

\begin{figure}[t!]
\centering
\includegraphics[width=0.48\textwidth, clip=true,trim= 0.15in 0.17in 0.15in 0.20in]{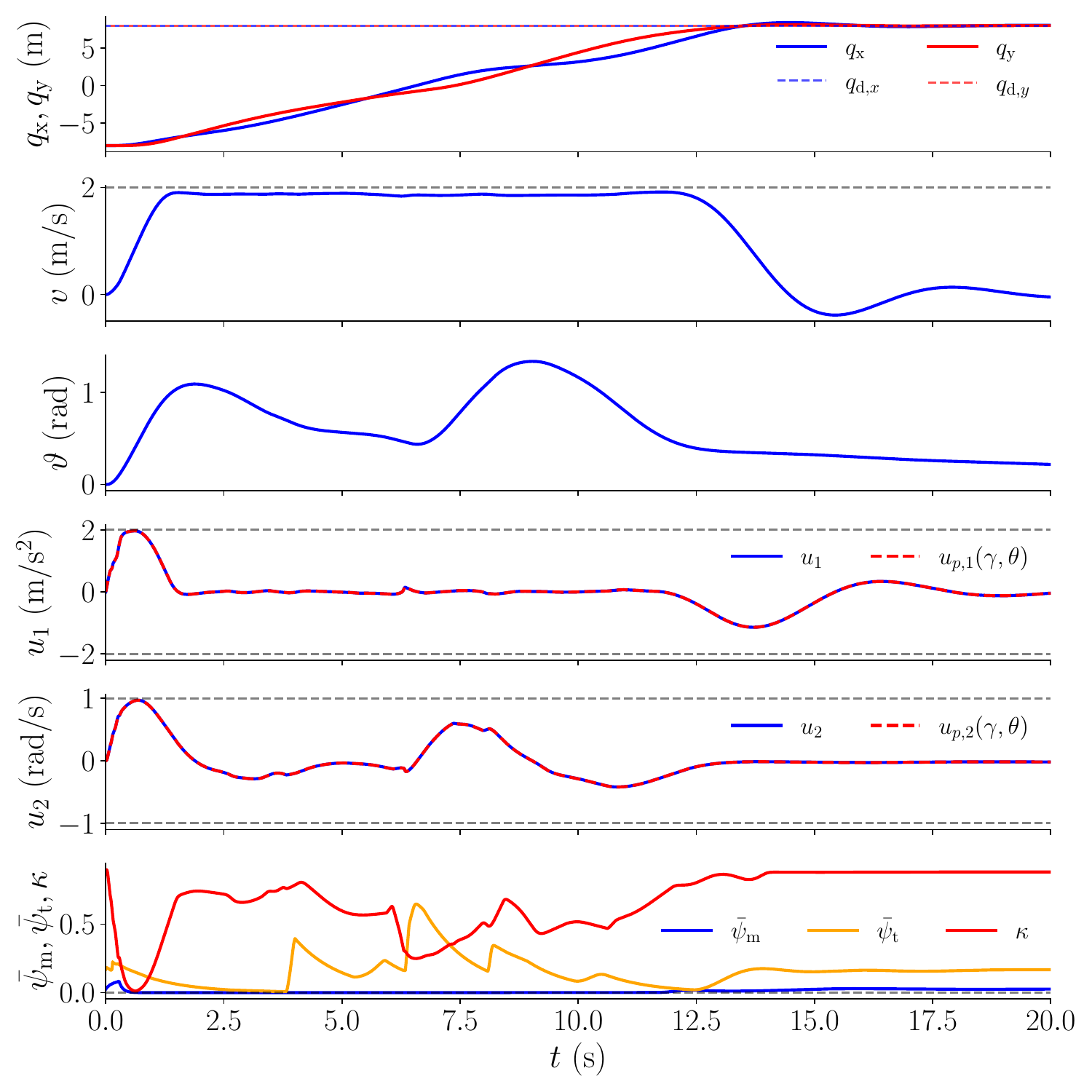}
\caption{Time histories of robot states $(q_\rmx, q_\rmy, v, \vartheta)$, control inputs $(u_1, u_2)$, and barrier functions $\bar\psi_\rmm$, $\bar\psi_\rmt$, and $k$ for $x_0 = [\,-8\; -8\; 0 \; 0 \,]^\top$ and $x_\rmd = [\,8\; 8 \; 0 \; 0\,]^\top$.}
\label{fig:robot_state}
\end{figure}

\begin{figure}[t!]
\centering
\includegraphics[width=0.48\textwidth, clip=true,trim= 0.15in 0.07in 0.1in 0.05in]{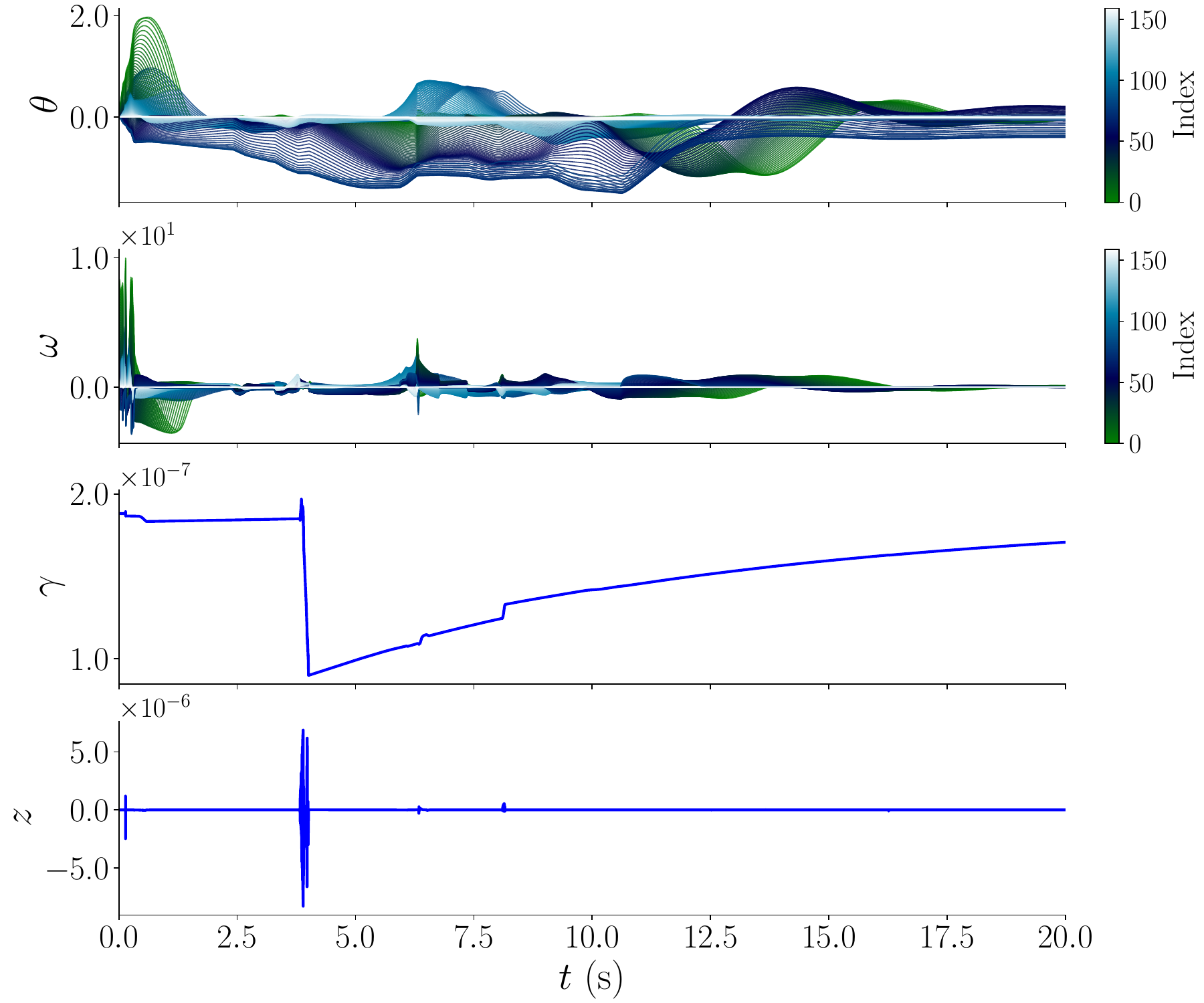}
\caption{Time histories of control parameters $\theta$, parameter rate $\omega$, planning-time shift $\gamma$, and planning-time shift rate $z$ for $x_0 = [\,-8\; -8\; 0 \; 0 \,]^\top$ and $x_\rmd = [\,8\; 8 \; 0 \; 0\,]^\top$.}
\label{fig:param}
\end{figure}

\Cref{fig:monte_carlo} highlights in black a closed-loop trajectory of FlowBarrier on a navigation task from $\bar{x}_0 = [\,-8\; -8\; 0\; 0\; 0_d\; 0\,]^\top$ to goal state $x_\rmd = [\,8\; 8 \; 0 \; 0\,]^\top$. \Cref{fig:robot_state} shows time histories of the robot states $(q_\rmx, q_\rmy, v, \vartheta)$, control inputs $(u_1, u_2)$, and barrier functions $\bar\psi_\rmm$, $\bar\psi_\rmt$, and $k$. The barrier functions remain nonnegative throughout the trajectory, confirming that $\bar x \in \bar\Psi$ for all $t \geq 0$.
The control inputs satisfy the input constraints $u_1(t) \in [-2, 2]$ and $u_2(t) \in [-1, 1]$, and the speed satisfies the constraint $v(t) \in [-2, 2]$ imposed by $b_{48}$ and $b_{49}$.
The relatively large weight $Q_u = 10^6 I_m$ encourages the optimal control $u$ to stay close to the instantaneous plan $u_\rmp(\gamma(t);\theta(t))$. 

\Cref{fig:param} shows the evolution of control parameters $\theta$, parameter rate $\omega$, planning-time shift $\gamma$, and planning-time shift rate $z$ during safe navigation through the obstacle field.
The planning-time shift $\gamma$ is near zero because the cost \eqref{eq:qp_cost} penalizes increasing $\gamma$ through the term $\lambda \hat{z}$ with $\lambda = 10^3$.
In this example, the control plan parameter $\theta \in \BBR^{160}$ provides sufficient degrees of freedom for the optimizer \eqref{eq:v_star_def} to satisfy the constraints without shrinking the prediction horizon.
In other words, the choice of large $\lambda$ makes it so that the horizon shrinks (i.e., $\gamma$ increases) only as a last resort if needed for feasibility.

For comparison, we present simulation results with alternative control approaches. Specifically, we compare the FlowBarrier method with nonlinear model predictive control (NMPC) \cite{mayne2000}, and a constrained iterative linear quadratic regulator (CiLQR) \cite{jallet2022constrained} paired with two different safety filters: a control barrier function filter (CiLQR-CBF) \cite{xiao2021high} and a backup control barrier function filter (CiLQR-BCBF) \cite{backupautomatic}.

The NMPC approach solves a finite-horizon optimal control problem that minimizes the cost functional \eqref{eq:def:J}, where $R$ and $W$ are given by \eqref{eq:cost_functions}, subject to the state constraint $h_\rms(x) \geq 0$ along the prediction horizon, the terminal constraint $h_\rmb(x(T)) \geq 0$, and the input constraints $u \in \SU$. The prediction horizon $T$ and discretization time step $T_\rmd$ are set equal to those used in the FlowBarrier method. Since the terminal constraint enforces that the predicted terminal state lies within the forward invariant backup safe set $\SC_\rmb$, recursive feasibility of the NMPC is guaranteed.

The CiLQR-CBF approach consists of two stages. First, a constrained iterative linear quadratic regulator is employed to compute a nominal control trajectory that minimizes the cost functional \eqref{eq:def:J} with $R$ and $W$ given by \eqref{eq:cost_functions}, where the state constraint $h_\rms(x) \geq 0$ and the input constraints $u \in \SU$ are enforced via an augmented Lagrangian method. The prediction horizon $T$ and discretization time step $T_\rmd$ are set equal to those used in the FlowBarrier method. The resulting nominal control serves as the desired control input for a CBF safety filter, which solves a quadratic program that minimally modifies the desired control to enforce safety and input constraints (see \cite{xiao2021high} for details).

The CiLQR-BCBF approach employs the same constrained iterative linear quadratic regulator to generate the desired control. However, instead of the standard CBF, a backup control barrier function (BCBF) method is used as the safety filter.

The BCBF uses a backup controller that drives the robot to rest and solves a quadratic program that minimally modifies the desired control while ensuring the predicted trajectory under the backup controller remains in the safe set along the prediction horizon and satisfies a terminal safety condition, subject to input constraints $u \in \SU$ (see \cite{backupautomatic} for details).

All simulations are performed in Python on a laptop computer with an Intel Core i9-14900HX CPU and 32 GB of RAM. All methods are implemented using the \texttt{CBFJAX} framework, which is built on JAX and provides automatic differentiation and just-in-time compilation. Within the \texttt{CBFJAX} framework, the NMPC problem is solved using \texttt{do-mpc} \cite{do-mpc} with \texttt{IPOPT} \cite{ipopt}, and the CiLQR problem is solved using \texttt{trajax} \cite{trajax}. All numerical ODE integration and adjoint computations are performed using \texttt{Diffrax} \cite{diffrax}, and quadratic programming problems are solved using \texttt{JaxOpt} \cite{jaxopt} with \texttt{OSQP} \cite{osqp}.

To compare methods across diverse conditions, we conduct $100$ navigation tasks constructed from a grid of initial and goal configurations. Specifically, $10$ initial states are uniformly sampled along the line from $x_0 = [\,-8\; -8\; 0\; 0\,]^\top$ to $[\,8\; -8\; 0\; \pi\,]^\top$, and $10$ goal states are uniformly sampled along the line from $x_\rmd = [\,-8\; 8\; 0\; 0\,]^\top$ to $[\,8\; 8\; 0\; 0\,]^\top$, yielding $100$ distinct navigation tasks by pairing each initial state with each goal state. Each simulation is run for $20\,\rms$. For fair comparison, all methods employ the same running cost $R$ and terminal cost $W$ given by \eqref{eq:cost_functions}, the same prediction horizon $T = 4\, \rms$ and discretization time step $T_\rmd = 0.05\, \rms$, with the control trajectory initialized to zero. Parameters shared across all methods are set to identical values, while method-specific parameters are individually selected to achieve best performance for each approach.

Trajectories are categorized as \emph{reached} if they successfully arrive at the goal with $\|x - x_\rmd\| \leq 0.5$ within $20\,\rms$, \emph{stuck} if they fail to make progress, or \emph{failed} if they violate safety constraints. \Cref{fig:monte_carlo} shows the resulting trajectories for all methods. FlowBarrier achieves $88$ reached, $12$ stuck, and $0$ failed. NMPC achieves competitive performance with $85$ reached, $15$ stuck, and $0$ failed. CiLQR-CBF and CiLQR-BCBF demonstrate lower success rates, achieving $37$ reached, $55$ stuck, and $8$ failed, and $44$ reached, $55$ stuck, and $0$ failed, respectively.

\Cref{fig:mc_box} presents detailed statistical comparisons of time to goal $\text{TTG} \triangleq \min\{\hat t : \text{for all } t \geq \hat t, \|x(t) - x_\rmd\| \leq 0.5\}$, cumulative cost $J_{\text{cum}} \triangleq \int_0^{20} R(x(t), u(t)) \, dt$, computation time, minimum barrier over time $\min_{t \in [0, 20]} h_\rms(x(t))$, minimum barrier over prediction horizon $\min_{t \in [0, 20], \tau \in [0, T]} h_\rms(\varphi(\tau; \bar x(t)))$, and prediction violations. FlowBarrier achieves competitive time to goal and cumulative cost while maintaining the lowest computation time among all methods. The minimum barrier over time remains nonnegative for FlowBarrier, NMPC, and CiLQR-BCBF across all trials, while CiLQR-CBF exhibits $8$ safety violations due to infeasibility of the CBF quadratic program. Most critically, FlowBarrier is the only method with zero prediction violations across all $100$ trials, demonstrating formal safety guarantees not only on the executed trajectory but also on the planned trajectory.

\begin{figure}[t!]
\centering
\includegraphics[width=0.48\textwidth, clip=true, trim=0.1in 0.1in 0.25in 0.1in]{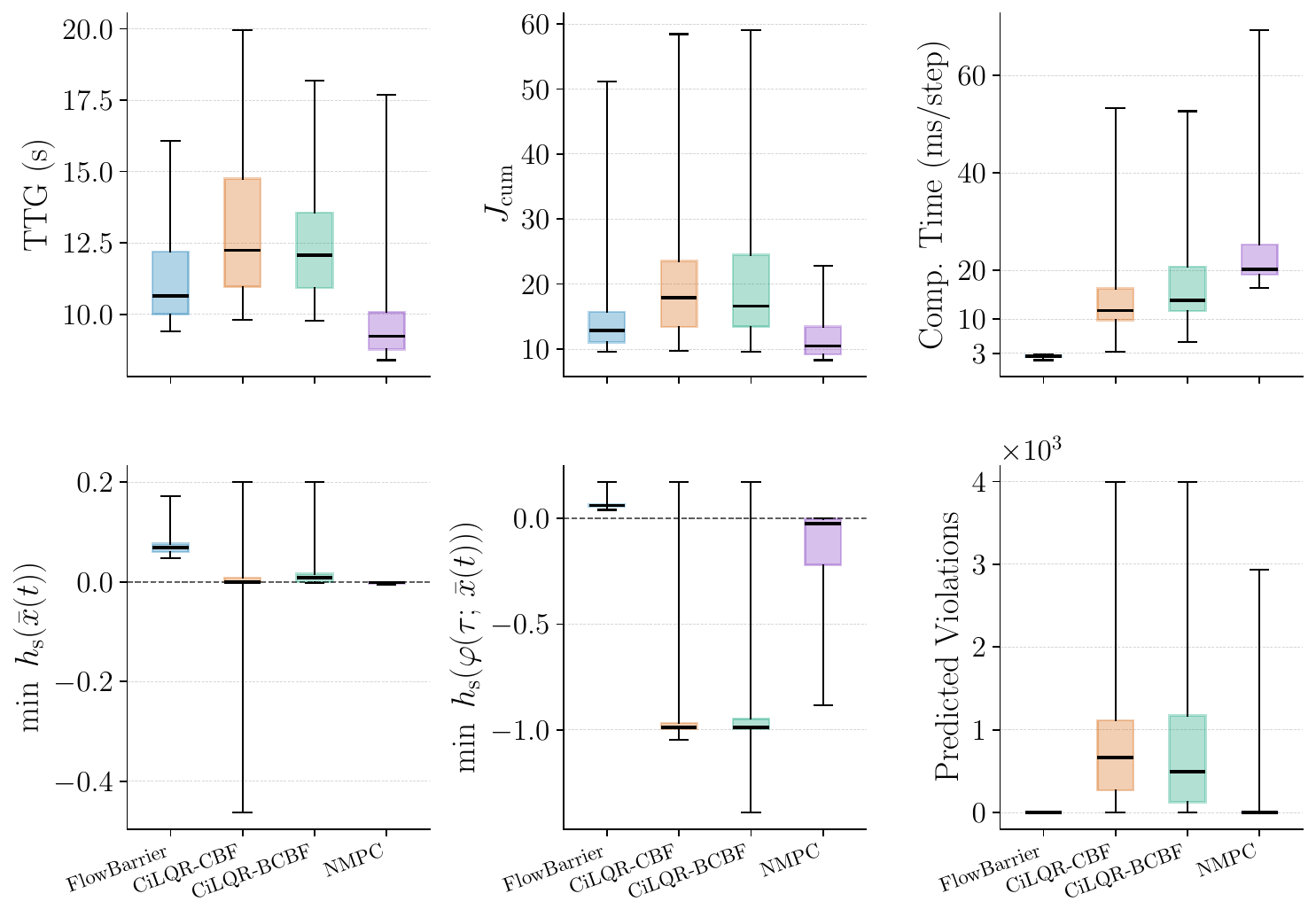}
\caption{Statistical comparison of performance metrics across 100 navigation trials for FlowBarrier, NMPC, CiLQR-CBF, and CiLQR-BCBF. Metrics include time to goal, cumulative cost, computation time, minimum barrier over time $\min_{t \in [0, 20]} h_\rms(x(t))$, minimum barrier over prediction horizon $\min_{t \in [0, 20], \tau \in [0, T]} h_\rms(\varphi(\tau; \bar x(t)))$, and prediction violations.}
\label{fig:mc_box}
\end{figure}

This low-dimensional nonholonomic ground robot is used to demonstrate FlowBarrier.
%For low-dimensional systems, Hamilton-Jacobi reachability analysis \cite{bansal2017hj,chen2018hj} can be used to compute the maximal reach-avoid set within $\SC_\rms$, and the resulting value function can be used as a safety certificate (see \cite{wiltz2025predictive}).
For low-dimensional systems, Hamilton-Jacobi reachability analysis \cite{bansal2017hj,chen2018hj} can be used to compute the maximal forward invariant subset of $\SC_\rms$, and the resulting value function is a safety certificate \cite{wiltz2025predictive}.
However, Hamilton-Jacobi reachability scales poorly with the dimension of the state.
In contrast, FlowBarrier scales linearly with $n$ as discussed in \Cref{sec:impl_remark}.
Hence, FlowBarrier can be used with higher-dimensional systems where reachability analysis may not be tractable.
Additionally, FlowBarrier unifies finite-horizon integral-cost optimization with safety certification.
%The comparisons in this section focus on real-time methods, and a direct comparison with Hamilton-Jacobi reachability is beyond the scope of this article.

\section{CONCLUDING REMARKS}

This article introduced P-CBFs, which generalize the CBF concept to a functional of a predicted flow under a parametrized control plan.
Thus, P-CBFs can be used for non-myopic real-time planning and control with safety constraints. 
However, candidate P-CBFs can be difficult to verify under control input constraints.
This article resolved the validity challenge by introducing a backup safe set and a planning-time shift that modulates the prediction horizon.
Together, the planning-time shift and backup safe set guarantee that the feasible set is nonempty.

Next, this article presented the safe optimal flow control method, where the real-time control $u$ and time evolution of $\theta$ and $\gamma$ are determined from a single convex optimization that is guaranteed to be feasible provided a backup safe set is known. 
The resulting safe optimal flow control not only guarantees safety of the state trajectory but also provides a safety certificate of the predicted flow over the entire prediction horizon at each instant of real time. 
The method unifies finite-horizon integral-cost optimization with safety certification in a single convex optimization whose feasibility under input constraints is guaranteed.

Notably, the safe-optimal-flow-control method does not require any explicit relative degree assumption related to the predictive-flow barrier functions $\bar\psi_\rmm$ or $\bar\psi_\rmt$. 
The only assumption related to relative degree is Assumption~\ref{ass:backup_set}, which implies that the backup CBF $h_\rmb$ has relative degree one. 
However, this can be directly extended to higher relative degree using standard higher-order CBF approaches \cite{xiao2021high}.

The feasibility and safety guarantees in this article rely on exact knowledge of the system and environment.
In practice, model uncertainty, external disturbances, and \textit{a priori} unknown and/or dynamic environments can compromise these guarantees.
Future work includes understanding robustness to model uncertainty and disturbances as well as extensions to \textit{a priori} unknown and/or dynamic environments.

\renewcommand{\appendixname}{APPENDIX}
\appendix

\section{PROOFS}\label{app:proofs}

\begin{proof}[Proof of Proposition~\ref{prop:directional_derivative_danskin}]
\indent Since $f$, $g$, and $h_\rms$ are continuously differentiable on $\BBR^n$, $u_\rmp$ is continuous on $[0,T] \times \BBR^d$, and $u_\rmp(\tau;\cdot)$ is continuously differentiable on $\BBR^d$, it follows that $h_\rms(\phi(\tau;x,\theta))$ is continuously differentiable on $[0,T] \times \BBR^n \times \BBR^d$.

To prove \ref{prop:directional_derivative_danskin.b}, let $(x_0,\theta_0) \in \BBR^n \times \BBR^d$ and $r > 0$, and define $\SB_r \triangleq \{(x,\theta) \in \BBR^n \times \BBR^d \colon \|(x,\theta) - (x_0,\theta_0)\| \leq r\}$.
Since $[0,T] \times \SB_r$ is compact and $h_\rms(\phi(\tau;x,\theta))$ is continuously differentiable on $[0,T] \times \SB_r$, it follows that
\begin{equation*}
M \triangleq \!\!\!\!\!\max_{(\tau,x,\theta) \in [0,T] \times \SB_r}\! \left \| h_\rms^\prime(\phi(\tau;x,\theta)) \begin{bmatrix} \frac{\partial \phi}{\partial x}(\tau;x,\theta) & \!\!\frac{\partial \phi}{\partial \theta}(\tau;x,\theta) \end{bmatrix} \right \|
\end{equation*}
exists. 
Since $\SB_r$ is convex, it follows that for all $\tau \in [0,T]$
and all $(x_1,\theta_1),(x_2,\theta_2) \in \SB_r$,
\begin{equation}\label{eq:lip_hs}
|h_\rms(\phi(\tau;x_1,\theta_1)) - h_\rms(\phi(\tau;x_2,\theta_2))| \leq M \|(x_1,\theta_1) - (x_2,\theta_2)\|.
\end{equation}
Next, it follows from \cref{eq:def:psi_m,eq:lip_hs} that for all $\tau \in [0,T]$,
\begin{align}
\psi_\rmm(x_1,\theta_1) &\leq h_\rms(\phi(\tau;x_1,\theta_1)) \nn \\
&\le h_\rms(\phi(\tau;x_2,\theta_2))\nn\\ &\qquad +|h_\rms(\phi(\tau;x_1,\theta_1))-h_\rms(\phi(\tau;x_2,\theta_2))|\nn\\
&\leq h_\rms(\phi(\tau;x_2,\theta_2)) + M \|(x_1,\theta_1) - (x_2,\theta_2)\|, \nn
\end{align}
and taking the minimum over $\tau \in [0,T]$ yields $\psi_\rmm(x_1,\theta_1) \leq \psi_\rmm(x_2,\theta_2) + M \|(x_1,\theta_1) - (x_2,\theta_2)\|$.
Similarly, $\psi_\rmm(x_2,\theta_2) \leq \psi_\rmm(x_1,\theta_1) + M \|(x_1,\theta_1) - (x_2,\theta_2)\|$.
Together, these inequalities imply $|\psi_\rmm(x_1,\theta_1) - \psi_\rmm(x_2,\theta_2)| \leq M \|(x_1,\theta_1) - (x_2,\theta_2)\|$.
Thus, $\psi_\rmm$ is locally Lipschitz on $\BBR^n \times \BBR^d$, which confirms \ref{prop:directional_derivative_danskin.b}.

To prove \ref{prop:directional_derivative_danskin.c}, since $h_\rms(\phi(\tau;x,\theta))$ is continuously differentiable and $[0,T]$ is compact, it follows from Danskin's theorem \cite[Theorem~1]{danskin2012theory} that
$\psi_\rmm$ is directionally differentiable with directional derivative \eqref{eq:dini_psi_formula_danskin}.
\end{proof}

\begin{proof}[Proof of Proposition~\ref{prop:backup_controller_exists}]
\indent Since for all $x\in\SC_\rmb$, $K_\rmb(x)$ is nonempty, it follows that $h_\rmb$ is a CBF for \eqref{eq:affine_control} on $\SC_\rmb$. 
Thus, Theorem~\ref{prop:vector_dcbf_invariance} implies that $\SC_\rmb$ is forward invariant with respect to \eqref{eq:affine_control} with $u=u_\rmb$.

To prove $u_\rmb$ is continuous, let $x_0 \in \SC_\rmb$, and define $c_\rmb(x,\hat{u}) \triangleq L_f h_\rmb(x) + L_g h_\rmb(x) \hat{u} + \alpha_\rmb(h_\rmb(x))$.
Assumption~\ref{ass:backup_set} implies that there exists $u_0 \in \SU$ such that $c_\rmb(x_0,u_0) > 0$.
Since $\SU$ is convex, $c_\rmb$ is continuous on $\SC_\rmb \times \SU$, and $c_\rmb(x,\hat{u})$ is affine in $\hat{u}$, it follows from \cite[Theorem~5.1]{borrelli2017predictive} that $K_\rmb$ is continuous at $x_0$ in the sense of \cite[Definition~5.3]{borrelli2017predictive}.
Since $u_\rmb(x) \in K_\rmb(x) \subseteq \SU$ and $\SU$ is compact, it follows that the set of minimizers is uniformly compact near $x_0$.
Since $\|\hat{u}\|^2$ is strictly convex, it follows that $u_\rmb(x_0)$ is the unique minimizer.
Since, in addition, $K_\rmb$ is continuous at $x_0$ and $\|\hat{u}\|^2$ is continuous on $\SC_\rmb \times K_\rmb(x_0)$, it follows from \cite[Theorem~5.3]{borrelli2017predictive} that $u_\rmb$ is continuous at $x_0$, which implies that $u_\rmb$ is continuous on $\SC_\rmb$.
\end{proof}

\begin{proof}[Proof of Proposition~\ref{prop:psi_lipschitz}]
\indent Since $f$, $g$, and $h_\rms$ are continuously differentiable on $\BBR^n$, and $u_\rmp$ is continuous on $[0,T] \times \BBR^d$ and, for all $\tau \in [0,T]$, $u_\rmp(\tau;\cdot)$ is continuously differentiable on $\BBR^d$, it follows that $h_\rms(\varphi(\tau;\bar{x}))$ is continuously differentiable on $\{(\tau, \bar{x}) \in [0,T] \times \BBR^{n} \times \BBR^d \times [0,T] \colon \tau \in [\gamma, T]\}$.

To prove \ref{prop:psi_lipschitz_c}, let $\bar{x}_\rme \in \BBR^{n} \times \BBR^d \times [0,T]$ and $r > 0$, and define $\bar\SB_r \triangleq \{\bar{x} \in \BBR^{n} \times \BBR^d \times [0,T] \colon \|\bar{x} - \bar{x}_\rme\| \leq r\}$ and $\SQ_r \triangleq \{(\tau,\bar{x}) \in [0,T] \times \bar\SB_r \colon \tau \in [\gamma, T]\}$.
Since $\SQ_r$ is compact and $h_\rms(\varphi(\tau;\bar{x}))$ is continuously differentiable on $\SQ_r$, it follows that
\begin{equation*}
M_{\bar{x}} \triangleq \max_{(\tau,\bar{x}) \in \SQ_r} \left \| h_\rms^\prime(\varphi(\tau;\bar{x})) \frac{\partial \varphi}{\partial \bar{x}}(\tau;\bar{x}) \right \|
\end{equation*}
exists. 
Since $\bar\SB_r$ is convex, it follows that for all $\tau \in [\max \{\gamma_1, \gamma_2\} , T]$ and all $\bar{x}_1, \bar{x}_2 \in \bar\SB_r$,
\begin{equation}\label{eq:lip_xbar}
|h_\rms(\varphi(\tau;\bar{x}_1)) - h_\rms(\varphi(\tau;\bar{x}_2))|
\leq M_{\bar{x}} \|\bar{x}_1 - \bar{x}_2\|.
\end{equation}
Similarly,
\begin{equation*}
M_\tau \triangleq \max_{(\tau,\bar{x}) \in \SQ_r} \left | h_\rms^\prime(\varphi(\tau;\bar{x})) F(\varphi(\tau;\bar{x}), u_\rmp(\tau;\theta)) \right |
\end{equation*}
exists.
Since $[\gamma,T]$ is convex, it follows that for all $\bar{x} \in \bar\SB_r$ and all $\tau_1, \tau_2 \in [\gamma, T]$,
\begin{equation}\label{eq:lip_tau}
|h_\rms(\varphi(\tau_1;\bar{x})) - h_\rms(\varphi(\tau_2;\bar{x}))|
\leq M_\tau |\tau_1 - \tau_2|.
\end{equation}

Let $\bar{x}_1, \bar{x}_2 \in \bar\SB_r$, where without loss of generality, $\gamma_1 \leq \gamma_2$.
Since $[\gamma_2, T] \subseteq [\gamma_1, T]$, it follows from \eqref{eq:psi_min} and \eqref{eq:lip_xbar} that for all $\tau \in [\gamma_2, T]$,
\begin{align}
\bar\psi_\rmm(\bar{x}_1) &\leq h_\rms(\varphi(\tau;\bar{x}_1)) \nn \\
&\leq h_\rms(\varphi(\tau;\bar{x}_2)) + M_{\bar{x}} \|\bar{x}_1 - \bar{x}_2\|, \nn
\end{align}
and taking the minimum over $\tau \in [\gamma_2, T]$ yields
\begin{equation}\label{eq:dir1_result}
\bar\psi_\rmm(\bar{x}_1) \leq \bar\psi_\rmm(\bar{x}_2) + M_{\bar{x}} \|\bar{x}_1 - \bar{x}_2\|.
\end{equation}

Next, let $\tau_1 \in \overline{\ST}(\bar{x}_1)$, which implies $\bar\psi_\rmm(\bar{x}_1) = h_\rms(\varphi(\tau_1;\bar{x}_1))$.  
We consider 2 cases: $\tau_1 \in [\gamma_2,T]$ and $\tau_1 \in [\gamma_1,\gamma_2)$.
First, consider $\tau_1 \in [\gamma_2,T]$, and it follows from 
\cref{eq:psi_min,eq:lip_xbar} that
\begin{align}
\bar\psi_\rmm(\bar{x}_2) &\leq h_\rms(\varphi(\tau_1;\bar{x}_2)) \nn \\
&\leq h_\rms(\varphi(\tau_1;\bar{x}_1)) + M_{\bar{x}} \|\bar{x}_1 - \bar{x}_2\| \nn \\
&= \bar\psi_\rmm(\bar{x}_1) + M_{\bar{x}} \|\bar{x}_1 - \bar{x}_2\|. \label{eq:dir2_result}
\end{align}
Next, consider $\tau_1 \in [\gamma_1,\gamma_2)$, and it follows from 
\cref{eq:psi_min,eq:lip_xbar,eq:lip_tau} that
\begin{align}
\bar\psi_\rmm(\bar{x}_2) &\leq h_\rms(\varphi(\gamma_2;\bar{x}_2)) \nn \\
&\leq h_\rms(\varphi(\gamma_2;\bar{x}_1)) + M_{\bar{x}} \|\bar{x}_1 - \bar{x}_2\| \nn \\
&\leq h_\rms(\varphi(\tau_1;\bar{x}_1)) + M_\tau | \gamma_2 - \tau_1 | + M_{\bar{x}} \|\bar{x}_1 - \bar{x}_2\| \nn \\
&\leq h_\rms(\varphi(\tau_1;\bar{x}_1)) + M_\tau | \gamma_2 - \gamma_1 | + M_{\bar{x}} \|\bar{x}_1 - \bar{x}_2\| \nn \\
&= \bar\psi_\rmm(\bar{x}_1) + M_\tau|\gamma_2 - \gamma_1| + M_{\bar{x}} \|\bar{x}_1 - \bar{x}_2\|\nn\\
&\le \bar\psi_\rmm(\bar{x}_1) + (M_\tau+ M_{\bar{x}}) \|\bar{x}_1 - \bar{x}_2\|. \label{eq:dir3_result}
\end{align}
Together, \cref{eq:dir1_result,eq:dir2_result,eq:dir3_result} imply $|\bar\psi_\rmm(\bar{x}_1) - \bar\psi_\rmm(\bar{x}_2)| \leq (M_\tau + M_{\bar{x}}) \|\bar{x}_1 - \bar{x}_2\|$.
Thus, $\bar\psi_\rmm$ is locally Lipschitz on $\BBR^{n} \times \BBR^d \times [0,T]$, which confirms \ref{prop:psi_lipschitz_c}.

To prove \ref{prop:psi_lipschitz_d}, note that for all $\bar{x} \in \{\bar{x} \in \bar\Psi \colon \gamma \neq T\}$, \eqref{eq:psi_min} is a minimization of $h_\rms(\varphi(\tau;\bar{x}))$ over $\tau \in [\gamma,T]$.
The Lagrangian associated with this minimization problem is
\begin{equation}\label{eq:lagrangian}
\mathcal{L}(\tau,\lambda_1,\lambda_2;\bar{x})
\triangleq h_\rms(\varphi(\tau;\bar{x})) - \lambda_1 c_1(\tau) - \lambda_2 c_2(\tau),
\end{equation}
where $c_1(\tau) \triangleq \tau - \gamma$, $c_2(\tau) \triangleq T - \tau$, and $\lambda_1, \lambda_2 \geq 0$.
Differentiating \eqref{eq:lagrangian} with respect to $\bar{x}$ yields
\begin{equation}\label{eq:grad_lagrangian}
\frac{\partial \mathcal{L}}{\partial \bar{x}}(\tau,\lambda_1,\lambda_2;\bar{x})
= h_\rms'(\varphi(\tau;\bar{x})) \frac{\partial \varphi}{\partial \bar{x}}(\tau;\bar{x})
+ \lambda_1 \frac{\partial \gamma}{\partial \bar{x}}.
\end{equation}
Since $\gamma < T$, the feasible set $[\gamma,T]$ is nonempty and compact, and $c_1$ and $c_2$ are not both equal to zero at any point in $[\gamma,T]$.
Since, in addition, $c_1^\prime \ne 0$ and $c_2^\prime \ne 0$ on $[\gamma,T]$, it follows that the linear independence constraint qualification is satisfied for all $\tau \in [\gamma,T]$.

Since, in addition, $h_\rms(\varphi(\tau;\bar{x}))$ is continuously differentiable on $\{(\tau, \bar{x}) \in \BBR \times \bar\Psi \colon \tau \in [\gamma, T]\}$, it follows from \cite[Corollary~4.4]{gauvin2009differential} that $\bar\psi_\rmm$ is directionally differentiable on $\{\bar{x} \in \bar\Psi \colon \gamma \neq T\}$, and for all $\bar{x} \in \{\bar{x} \in \bar\Psi \colon \gamma \neq T\}$,
\begin{equation}\label{eq:dini_from_corollary}
D_\nu\bar\psi_\rmm(\bar{x}) = \min_{\tau \in \overline{\ST}(\bar{x})}
\frac{\partial \mathcal{L}}{\partial \bar{x}}(\tau,\lambda_1^*(\tau),\lambda_2^*(\tau);\bar{x}) \, \nu,
\end{equation}
where, for all $\tau \in \overline{\ST}(\bar{x})$, $\lambda_1^*(\tau), \lambda_2^*(\tau) \geq 0$ are the unique multipliers satisfying $\lambda_1^*(\tau) \, c_1(\tau) = 0$, $\lambda_2^*(\tau) \, c_2(\tau) = 0$, and $\frac{\partial \mathcal{L}}{\partial \tau}(\tau,\lambda_1^*(\tau),\lambda_2^*(\tau);\bar{x}) = 0$.
Since, for all $\tau \in \overline{\ST}(\bar{x})$, $\frac{\partial \mathcal{L}}{\partial \tau}(\tau,\lambda_1^*(\tau),\lambda_2^*(\tau);\bar{x}) = 0$, it follows that
\begin{equation}\label{eq:lambda1_from_stationarity}
\lambda_1^*(\tau) = \lambda_2^*(\tau)
+ h_\rms'(\varphi(\tau;\bar{x})) \frac{\partial \varphi}{\partial \tau}(\tau;\bar{x}).
\end{equation}
We consider 2 cases: $\tau = \gamma$ and $\tau \in (\gamma, T]$.
First, consider $\tau = \gamma$. 
Since $c_2(\tau) > 0$, it follows that $\lambda_2^*(\tau) = 0$, and \eqref{eq:lambda1_from_stationarity} yields $\lambda_1^*(\tau) = h_\rms'(\varphi(\tau;\bar{x})) \frac{\partial \varphi}{\partial \tau}(\tau;\bar{x})$.
Next, consider $\tau \in (\gamma, T]$.
Since $c_1(\tau) > 0$, it follows that $\lambda_1^*(\tau) = 0$.
Hence, for all $\tau \in \overline{\ST}(\bar{x})$,
\begin{equation}\label{eq:lambda1_indicator}
\lambda_1^*(\tau) = \mathbf{1}_{\{\tau = \gamma\}}
h_\rms'(\varphi(\tau;\bar{x})) \frac{\partial \varphi}{\partial \tau}(\tau;\bar{x}).
\end{equation}
Substituting \cref{eq:lambda1_indicator,eq:grad_lagrangian} into \eqref{eq:dini_from_corollary} yields \eqref{eq:dini_psi_formula}.
\end{proof}

\bibliographystyle{ieeetr}
\bibliography{SafeFlow} 
\end{document}